\documentclass{article}

\usepackage{amsmath}
\usepackage{amsthm}
\usepackage{amssymb}
\usepackage{mathtools}
\usepackage[boxed, vlined]{algorithm2e}
\usepackage{cite}
\usepackage{tabularx}
\usepackage[T1]{fontenc} 
\usepackage[bookmarks, linktocpage, pagebackref, colorlinks = true, linkcolor = blue, urlcolor = blue, citecolor = red]{hyperref}

\newtheorem{theorem}{Theorem}
\newtheorem{lemma}[theorem]{Lemma}

\newtheorem{claim}{Claim}
\newtheorem{definition}{Definition}

\providecommand{\keywords}[1]{\vspace{\baselineskip}\noindent\textbf{\textit{Keywords:}} #1}

\allowdisplaybreaks

\begin{document}
\title{On Differentially Private Online Collaborative Recommendation Systems}
\author{Seth Gilbert \and Xiao Liu \and Haifeng Yu}
\date{}
\maketitle

\begin{abstract}
In collaborative recommendation systems, privacy may be compromised, as users' opinions are used to generate recommendations for others. In this paper, we consider an online collaborative recommendation system, and we measure users' privacy in terms of the standard differential privacy. We give the first quantitative analysis of the trade-offs between recommendation quality and users' privacy in such a system by showing a lower bound on the best achievable privacy for any non-trivial algorithm, and proposing a near-optimal algorithm. From our results, we find that there is actually little trade-off between recommendation quality and privacy for any non-trivial algorithm. Our results also identify the key parameters that determine the best achievable privacy.

\end{abstract}

\section{Introduction}
In this paper we consider an \emph{online collaborative recommendation system} that attempts to predict which objects its users will like. Imagine, for example, a news website which publishes articles every day. When a user enjoys an article, he/she votes on the article (e.g., upvotes it, likes it, +1s it, etc). Users can also ask the system for a recommendation, i.e., to suggest an article that they might like. After reading the recommended article, the user gives the system feedback on the recommendation so that it can improve its recommendation quality. In~this paper, we work with a simplified, abstract version of this very common paradigm.

Due to the way it works, a collaborative recommendation system has the risks of leaking its users' privacy. Clearly, there are trade-offs between recommendation quality and privacy: a system that gives completely random recommendations certainly leaks no one's privacy, but it is also useless; in contrast, a recommendation system that gives high quality recommendations has to make ``full use'' of its users' data, which is more prone to privacy leakage.

In this paper, we adopt $\epsilon$-differential privacy~\cite{dwork2006calibrating} as our formal definition of privacy, and we give the first quantitative analysis of these trade-offs for online collaborative recommendation systems. Prior to this paper, the topic of differentially private recommendation systems has primarily been examined under \emph{offline matrix} models~\cite{mcsherry2009differentially, chow2012practical, xin2014controlling, hardt2012beating, chaudhuri2013near, hardt2013beyond, kapralov2013differentially}. From the theoretical perspective, our recommendation model can be viewed as a variant of an \emph{online learning} problem. Currently, there are only a limited number of existing papers on differentially private online learning~\cite{dwork2010differential*, jain2011differentially, thakurta2013nearly}, and their privacy models do not fit the recommendation problem (see Section~\ref{section:related-work} for more details).

We first study the best achievable privacy for a fixed recommendation quality by showing a near-tight lower bound on the privacy parameter $\epsilon$ (smaller $\epsilon$ means better privacy). For example, if we were to guarantee a trivial recommendation quality only, then we can achieve ``perfect privacy'' (i.e., $\epsilon = 0$) by ignoring users' opinions on objects and recommending randomly. As we set better and better target recommendation quality, it might be expected that the best achievable $\epsilon$ smoothly gets larger and larger. However, we show that the transition is sharp: although $\epsilon = 0$ is achievable for the trivial recommendation quality, the lower bound of $\epsilon$ rises to a certain level as long as non-trivial recommendation quality is to be guaranteed, and it remains essentially the same (up to a logarithmic factor) as the target recommendation quality increases.

We then propose a novel $\epsilon$-differentially private algorithm. Our algorithm's $\epsilon$ is within a logarithmic factor to the aforementioned lower bound, and meanwhile its recommendation quality is also near-optimal up to a logarithmic factor, even when compared to algorithms providing no privacy guarantee.

Our near matching results surprisingly imply that there are actually little trade-offs between recommendation quality and privacy --- an inherent ``amount of privacy'' (up to a logarithmic factor) must be ``leaked'' for \emph{any} algorithm with non-trivial recommendation quality. Our results also identify the key parameters that fundamentally determine the best achievable recommendation quality and privacy. We provide more details about our results in Section~\ref{section:results}.

\section{Model and Problem Statement}
\subsection{Recommendation System Model}
We now describe the model in more detail, abstracting away some of the complications in the scenario above in order to focus on the fundamental trade-offs.

We consider an online collaborative recommendation system that contains \emph{voters}, \emph{clients} and \emph{objects}, and it repeatedly recommends objects to clients based on voters' opinions on objects. A voter/client either \emph{likes} or \emph{dislikes} an object. Voters submit their opinions on objects to the system in the form of \emph{votes}, where a vote by voter $i$ on object $j$ indicates that voter $i$ likes object $j$; clients receive recommendations from the system and provide \emph{feedback} to the system which tells whether they like the recommended objects or not. Since every client has his/her own personalized preferences, the system will serve each client separately.

We now describe how the model operates for a particular client $C$. The system runs for $T$ \emph{rounds}. In each round $t \in \{1, \dotsc, T\}$, a set of $m$ \emph{new} candidate objects arrives in the system, out of which the client $C$ likes at least one of them. We assume that $m$ is a constant, and totally the system has $mT$ objects over all the $T$ rounds. Let $\mathcal{U}$ denote the set of all the voters, and $\mathcal{B}_t$ denote the set of candidate objects in the $t$th round. After $\mathcal{B}_t$ arrives, each voter $i \in \mathcal{U}$ votes on one object in $\mathcal{B}_t$; the system then recommends one object $b_t \in \mathcal{B}_t$ to the client $C$ (based on the voters' votes and the previous execution history), and $C$ responses the system with his/her feedback which tells whether he/she likes $b_t$ or not. The system proceeds into the next round after that.

We measure the recommendation quality by \emph{loss}, which is defined as the number of objects that the algorithm recommends to the client $C$ but $C$ dislikes.

A client $C$ is fully characterized by specifying $C$'s preferences on every object. However, in a recommendation system, whether a client $C$ likes an object $j$ or not is unknown until the system has recommended $j$ to $C$ and gotten the feedback.

We denote the votes of all the voters in $\mathcal{U}$ by $\mathcal{V}\langle\mathcal{U}\rangle$, and we call $\mathcal{V}\langle\mathcal{U}\rangle$ the \emph{voting pattern of $\mathcal{U}$}, or simply a \emph{voting pattern} when $\mathcal{U}$ is clear from the context. Given a client $C$ and a voting pattern $\mathcal{V}\langle\mathcal{U}\rangle$, a (randomized) recommendation algorithm $\mathcal{A}$ maps the pair $(C, \mathcal{V}\langle\mathcal{U}\rangle)$ to a (random) sequence of objects in $\mathcal{B}_1 \times \dotsb \times \mathcal{B}_T$. We call a particular sequence in $\mathcal{B}_1 \times \dotsb \times \mathcal{B}_T$ a \emph{recommendation sequence}.

\subsection{Differential Privacy in Recommendation Systems}\label{section:model-differential-privacy}
Voters' votes are assumed to be securely stored by the system, which are not accessible from the public. Nevertheless, a curious client may still try to infer voters' votes by analyzing the recommendation results. In this paper, we adopt \emph{differential privacy}~\cite{dwork2006calibrating} as our definition of privacy. Roughly speaking, differential privacy protects privacy by ensuring that the outputs are ``similar'' for two voting patterns $\mathcal{V}\langle\mathcal{U}\rangle$ and $\mathcal{V}\langle\mathcal{U}'\rangle$ if they differ by one voter. Such a pair of voting patterns are called \emph{adjacent voting patterns}, and they are formally defined as:
\begin{definition}[Adjacent Voting Patterns]\label{definition:adjacent-voting-pattern}
Two voting patterns $\mathcal{V}\langle\mathcal{U}\rangle$ and $\mathcal{V}\langle\mathcal{U}'\rangle$ are adjacent voting patterns iff i) $|\mathcal{U} \,\triangle\, \mathcal{U}'| = 1$, and ii) for any voter $i \in \mathcal{U} \cap \mathcal{U}'$ and in any round $t \in \{1, \dotsc, T\}$, $i$ always votes on the same object in both $\mathcal{V}\langle\mathcal{U}\rangle$ and $\mathcal{V}\langle\mathcal{U}'\rangle$.
\end{definition}
Generalizing Definition~\ref{definition:adjacent-voting-pattern}, we say that two voting patterns $\mathcal{V}\langle\mathcal{U}\rangle$ and $\mathcal{V}\langle\mathcal{U}'\rangle$ are \emph{$k$-step adjacent}, if there exists a sequence of $k + 1$ voting patterns $\mathcal{V}\langle\mathcal{U}_0\rangle = \mathcal{V}\langle\mathcal{U}\rangle, \mathcal{V}\langle\mathcal{U}_1\rangle, \dotsc, \mathcal{V}\langle\mathcal{U}_{k - 1}\rangle, \mathcal{V}\langle\mathcal{U}_k\rangle = \mathcal{V}\langle\mathcal{U}'\rangle$ such that $\mathcal{V}\langle\mathcal{U}_\ell\rangle$ and $\mathcal{V}\langle\mathcal{U}_{\ell + 1}\rangle$ are adjacent for any $\ell = 0, \dotsc, k - 1$.

Having defined adjacent voting patterns, we can then apply the standard differential privacy in \cite{dwork2006calibrating} to our setting:
\begin{definition}[$\epsilon$-Differential Privacy]\label{definition:differential-privacy}
A recommendation algorithm $\mathcal{A}$ preserves $\epsilon$-differential privacy if $\Pr[\mathcal{A}(C, \mathcal{V}\langle\mathcal{U}\rangle) \in S] \le e^{\epsilon}\Pr[\mathcal{A}(C, \mathcal{V}\langle\mathcal{U}'\rangle) \in S]$ for any client $C$, any pair of adjacent voting patterns $\mathcal{V}\langle\mathcal{U}\rangle, \mathcal{V}\langle\mathcal{U}'\rangle$, and any subset $S \subseteq \mathcal{B}_1 \times \dotsb \times \mathcal{B}_T$, where the probabilities are over $\mathcal{A}$'s coin flips.
\end{definition}

\subsection{Attack Model, Power of the Adversary}
As indicated by Definition~\ref{definition:adjacent-voting-pattern} and Definition~\ref{definition:differential-privacy}, we protect voters' privacy against the client. We do not need to protect the client's privacy because voters receive nothing from the system.

Our research goal is to study the theoretical hardness of the aforementioned recommendation problem, therefore we assume that there is an adversary with unlimited computational power who controls how the voters vote and which objects the client likes. The adversary tries to compromise our algorithm's loss/privacy by feeding the algorithm with ``bad'' inputs. From the perspective of game theory, our recommendation model can be viewed as a repeated game between the algorithm, who chooses the objects to recommend, and the adversary, who chooses the client's preferences on objects and the voting pattern. For our lower bounds, we consider an \emph{oblivious adversary} that chooses the client's preferences on objects and the voting patterns in advance; for our upper bounds, we consider an \emph{adaptive adversary} whose choice in time $t$ can depend on the execution history prior to time $t$. By doing so, our results are only strengthened.

\subsection{Notations}
Next we introduce some notations that characterize the system. Some of them are also the key parameters that determine the best achievable loss/privacy.

\paragraph{The client's diversity of preferences.} A client $C$'s \emph{diversity of preferences} $D_C$ is defined to be the number of rounds in which $C$ likes more than one objects.

\paragraph{The client's peers.} Inherently, a collaborative recommendation system can achieve small loss only if some voters have similar preferences to the client. Let the \emph{distance} between a client $C$ and a voter $i$ be the total number of objects that are voted on by $i$ but are disliked by $C$. Given a radius parameter $R \in \{0, \dotsc, T\}$, we define a voter $i$ to be a client $C$'s \emph{peer} if their distance is within $R$. Given a client $C$, a voting pattern $\mathcal{V}\langle\mathcal{U}\rangle$ and a radius parameter $R$, we can count the number of $C$'s peers in $\mathcal{U}$, and we denote it by $P_{C, \mathcal{V}\langle\mathcal{U}\rangle, R}$.

\paragraph{Other notations.} We define $n$ to be an upper bound of $|\mathcal{U}|$ (i.e., the number of voters), $D$ to be an upper bound of $D_C$ (i.e., the client's diversity of preferences), and $P$ to be a lower bound of $P_{C, \mathcal{V}\langle\mathcal{U}\rangle, R}$ (i.e., the number of the client's peers). The reader may wonder why these parameters are defined as upper/lower bounds. The purpose is to give a succinct presentation. Take $n$ as an example: since differential privacy needs to consider two voting patterns with different numbers of voters, if we define $n$ as the number of voters, it would be unclear which voting pattern we are referring to. The reader can verify that by choosing the right directions for the parameters (e.g., we define $n$ to be an upper bound, and $P$ to be a lower bound), our definition does not weaken our results.

In general, we consider a large system that consists of many voters, many objects (over all the rounds), and runs for a long time. That is, $n$ and $T$ can be very large. In this paper, we also impose a (quite loose) requirement that $n = O(\mathrm{poly}(T))$, i.e., $n$ is not super large compared to $T$.

In reality, a client shall find more peers as more voters join the system. Otherwise, the client has an ``esoteric tastes'' and it is inherently hard for any collaborative system to help him/her. Thus, in this paper, we consider the case that $P \ge 6m$, i.e., the client has at least a constant number of peers.

Table~\ref{table:notations} summarizes the key notations in this paper.
\begin{table}[t]
	\centering
	\renewcommand\arraystretch{1.15}
	\caption{List of key notations.}\label{table:notations}.
	\begin{tabularx}{\textwidth}{cX}
		\hline
		Notation & Description\\
		\hline
		$T$ & the number of rounds\\
		$\mathcal{B}_t$ & the set of candidate objects in round $t$\\
		$m$ & the constant $m$ is the number of candidate objects in each round, i.e., $m = |\mathcal{B}_t|$ and totally the system has $mT$ objects\\
		$C$ & a client\\
		$\mathcal{U}$ & a set of voters\\
		$\mathcal{V}\langle\mathcal{U}\rangle$ & the voting pattern of $\mathcal{U}$\\
		$n$ & $n$ is an upper bound of the number of voters\\
		$R$ & $R \in \{0, \dotsc, T\}$ is a parameter that defines the radius of the client's peer group, i.e., a voter is a client's peer if their distance is within $R$\\
		$P_{C, \mathcal{V}\langle\mathcal{U}\rangle, R}$ & the number of the client $C$'s peers in the voting pattern $\mathcal{V}\langle\mathcal{U}\rangle$, given the radius parameter is $R$\\
		$P$ & $P$ is a lower bound of the number of the client's peers\\
		$D_C$ & the client $C$'s diversity of preferences, i.e., the client $C$ likes more than one objects in $D_C$ rounds\\
		$D$ & $D \in \{0, \dotsc, T\}$ is an upper bound of the clients' diversities of preferences\\
		\hline
	\end{tabularx}
\end{table}

\subsection{Loss/Privacy Goal}
In this paper, we consider the \emph{worst-case expected loss} of the algorithm, that is, we aim to bound the algorithm's expected loss for any client $C$ and any voting pattern $\mathcal{V}\langle\mathcal{U}\rangle$ such that $|\mathcal{U}| \le n$, $D_C \le D$ and $P_{C, \mathcal{V}\langle\mathcal{U}\rangle, R} \ge P$. Notice that $O(T)$ loss can be trivially achieved by ignoring voters' votes and recommending objects randomly. However, such an algorithm is useless, and hence we consider the more interesting case when sub-linear loss (in terms of $T$) is to be guaranteed. It can be shown that the worst-case expected loss is $\Omega(R)$ for any algorithm (Theorem~\ref{theorem:lower-bound-loss}). Therefore, sub-linear loss is achievable only when $R$ is sub-linear. In this paper, we focus on the case when $R = O(T^\eta)$ for some constant $\eta < 1$.\footnote{Technically, the assumptions that $n = O(\mathrm{polylog}(T))$, $P \ge 6m$ and $R = O(T^{\eta})$ are only for showing the near-optimality of our lower bound. Our lower bound itself remains to hold without these assumptions.}

For the privacy, we aim to preserve $\epsilon$-differential privacy. We study the best achievable $\epsilon$-differential privacy for any given target loss.

\section{Related Work}\label{section:related-work}
\paragraph{Recommendation Systems and Online Learning.} The research on recommendation systems has a long history~\cite{adomavicius2005toward, su2009survey}. A classic recommendation model is the \emph{offline matrix-based} model, in which the user-object relation is represented by a matrix. In this paper, we consider a very different \emph{online} recommendation model. From the theoretical perspective, our model can be viewed as a variant of the ``Prediction with Expert Advice'' (PEA) problem in online learning~\cite{cesa2006prediction}. Such an approach that models the recommendation systems as online learning problems has been adopted by other researchers as well, e.g., in \cite{awerbuch2007online, lee2001collaborative, nakamura1998collaborative, resnick2007influence, yu2009dsybil}.

\paragraph{Differential Privacy.} There has been abundant research on differential privacy~\cite{dwork2008differential, dwork2009differential, dwork2011firm, dwork2010differential}. Much of the early research focused on answering a single query on a dataset. Progress on answering multiple queries with non-trivial errors was made later on, for both offline settings~\cite{blum2008learning, dwork2009complexity, hardt2010multiplicative, roth2010interactive} (where the input is available\\ in advance), and online settings~\cite{dwork2010differential*, chan2011private, bolot2013private, chan2012differentially, kellaris2014differentially, jain2011differentially, thakurta2013nearly} (where the input continuously comes). We will introduce the work on \emph{differentially private online learning} in \cite{dwork2010differential*, jain2011differentially, thakurta2013nearly} with more details soon after, as they are most related to this paper.

\paragraph{Protecting Privacy in Recommendation Systems.} People are well aware of the privacy risks in collaborative recommendation systems. Two recent attacks were demonstrated in \cite{narayanan2008robust} (which de-anonymized a dataset published by Netflix) and \cite{calandrino2011you} (which inferred users' historical data by combining passive observation of a recommendation system with auxiliary information). 

Many of the existing privacy-preserving recommendation systems adopted~privacy notions other than differential privacy (e.g., \cite{berkovsky2007enhancing, canny2002collaborative, canny2002collaborative*, polat2003privacy, polat2005svd, shokri2009preserving}). For studies on \emph{differentially private} recommendation systems, prior to our paper, most of them were for \emph{offline matrix-based} models. Some experimentally studied the empirical trade-offs between loss and privacy (e.g.,~\cite{mcsherry2009differentially, chow2012practical, xin2014controlling}); the others focused on techniques that manipulate matrices in privacy-preserving ways (e.g.,~\cite{hardt2012beating, chaudhuri2013near, hardt2013beyond, kapralov2013differentially}).

\paragraph{Differentially Private Online Learning.} This paper is most related to \emph{differentially private online learning}, as our recommendation model is a variant of the PEA problem in online learning. Currently, only a limited number of studies have been done on this area~\cite{dwork2010differential*, jain2011differentially, thakurta2013nearly}. In \cite{dwork2010differential*}, Dwork et~al.\ proposed a differentially private algorithm for the PEA problem by plugging privacy-preserving online counters into ``Follow the Perturbed Leader'' algorithm~\cite{kalai2005efficient}. In \cite{jain2011differentially} and \cite{thakurta2013nearly}, differential privacy was considered under a more general online learning model called ``Online Convex Programming.''

Despite the similarity between our recommendation model and the learning models in \cite{dwork2010differential*, jain2011differentially, thakurta2013nearly}, there is an important difference. Since their research is not for recommendation systems, they considered somewhat different notions of privacy from ours. Roughly speaking, if interpreting their models as recommendation problems, then their privacy goal is to ensure that each \emph{voter} is ``followed'' with similar probabilities when running the algorithm with two adjacent voting patterns. Such a guarantee is not sufficient for a recommendation system. For example, an algorithm that always ``follows'' voter Alice is perfectly private in terms of their privacy definition, but completely discloses Alice's private votes.\footnote{On the other hand, our privacy definition does not imply their definitions either. Therefore these two types of privacy models are incomparable.} Besides the difference in privacy definition, we provide both lower bound and upper bound results, while \cite{dwork2010differential*, jain2011differentially, thakurta2013nearly} only have upper bound results.

\section{Our Results and Contributions}\label{section:results}
\paragraph{Main results.} Our first result is a lower bound on the best achievable privacy:
\begin{theorem}\label{theorem:lower-bound-privacy}
For any recommendation algorithm that guarantees $L = O(T^\beta)$ worst-case expected loss ($\beta < 1$ is a constant) and preserves $\epsilon$-differential privacy, $\epsilon = \Omega(\frac{1}{P}(D + R + \log\frac{T}{L})) = \Omega(\frac{1}{P}(D + R + \log{T}))$, even for an oblivious adversary.
\end{theorem}
Our second result is a near-optimal algorithm (the p-REC algorithm in Section~\ref{section:p-REC}):
\begin{theorem}\label{theorem:p-REC-loss-and-privacy}
The p-REC algorithm guarantees $O((R + 1)\log\frac{n}{P})$ worst-case expected loss, and it preserves $O(\frac{1}{P}(D + R + 1)\log\frac{T}{R + 1}) $-differential privacy, even for an adaptive adversary.
\end{theorem}
It can be shown that the worst-case expected loss is $\Omega(R + \log\frac{n}{P})$ even for algorithms with no privacy guarantee (Theorem~\ref{theorem:lower-bound-loss}). Thus, p-REC's worst-case expected loss is within a logarithmic factor to the optimal. Recall that $R = O(T^\eta)$ for a constant $\eta < 1$ and $\log{n} = O(\log{T})$, hence p-REC's worst-case expected loss is within $O(T^\beta)$ for some constant $\beta < 1$ too. Then, by Theorem~\ref{theorem:lower-bound-privacy}, p-REC's privacy is also within a logarithmic factor to the optimal.

\paragraph{Discussion of our results.} Theorem~\ref{theorem:lower-bound-privacy} shows that a minimal amount of ``privacy leakage'' is inevitable, even for the fairly weak $O(T^\beta)$ target loss.

Moreover, unlike many other systems in which the utility downgrades linear to the privacy parameter $\epsilon$, the loss in an online recommendation system is~much more sensitive to $\epsilon$: according to Theorem~\ref{theorem:p-REC-loss-and-privacy}, we can achieve near-optimal loss for an $\epsilon = O(\frac{1}{P}(D + R + 1)\log\frac{T}{R + 1})$; meanwhile, only trivial loss is achievable for just a slightly smaller $\epsilon = o(\frac{1}{P}(D + R + \log{T}))$. In other words, the trade-offs between loss and privacy are rather little --- the best achievable $\epsilon$ is essentially the same (up to a logarithmic factor) for \emph{all} the algorithms with $O(T^\beta)$ worst-case expected loss.\footnote{This statement actually holds for all the algorithms with $o(T)$ loss. In Theorem~\ref{theorem:lower-bound-privacy}, we choose $O(T^\beta)$ target loss to get a clean expression for the lower bound on $\epsilon$, and a similar (but messier) lower bound on $\epsilon$ holds for $o(T)$ target loss too.} For this reason, instead of designing an algorithm that has a tunable privacy parameter $\epsilon$, we directly propose the p-REC algorithm that simultaneously guarantees both near-optimal loss and privacy.

From our results, we identify the key parameters $D$, $P$ and $R$ that determine the best achievable loss and/or privacy.

The parameter $R$ characterizes the correlation between the client and the voters, and it is not surprised that the best achievable loss is inherently limited by $R$, because a basic assumption for any collaborative system is the existence of correlation in the data (e.g., the low-rank assumption in matrix-based recommendation systems), and the system works by exploring/exploiting the correlation.

We notice that a larger $P$ gives better privacy. This is consistent with our intuition, as an individual's privacy is obtained by hiding oneself in a population.

We also notice that the best achievable privacy linearly depends on the client's diversity of preferences $D$ and the radius parameter $R$. The parameter $D$ looks to be unnatural at the first sight, and no prior research on recommendation systems has studied it. The reason might be that most of the prior research focused on the loss, and $D$ has no impact on the loss (the loss should only be smaller if a client likes more objects). Nevertheless, in this paper, we discover that $D$ is one of the fundamental parameters that determine the best achievable privacy. We provide an intuitive explanation of $\epsilon$'s linear dependence on $D$ and $R$ with an illustrative example in Section~\ref{section:lower-bound-general-setting}.



\section{Preliminaries}
Let $\mathcal{P}$ and $\mathcal{Q}$ be two distributions over sample space $\Omega$. The \emph{relative entropy} between $\mathcal{P}$ and $\mathcal{Q}$ is defined as $\sum_{\omega \in \Omega}\mathcal{P}(\omega)\ln\frac{\mathcal{P}(\omega)}{\mathcal{Q}(\omega)}$, where $\mathcal{P}(\omega)$ and $\mathcal{Q}(\omega)$ is the probability of $\omega$ in $\mathcal{P}$ and $\mathcal{Q}$, respectively. We adopt the conventions that $0\log{\frac{0}{0}} = 0$, $0\log{\frac{0}{x}} = 0$ for any $x > 0$ and $x\log{\frac{x}{0}} = \infty$ for any $x > 0$. It is well known that relative entropy is always non-negative for any distributions $\mathcal{P}$ and $\mathcal{Q}$~\cite{kullback1968information}.

In this paper, we often simultaneously discuss two executions $\mathcal{A}(C, \mathcal{V}\langle\mathcal{U}\rangle)$ and $\mathcal{A}(C, \mathcal{V}\langle\mathcal{U}'\rangle)$ for some algorithm $\mathcal{A}$, some client $C$ and two voting patterns $\mathcal{V}\langle\mathcal{U}\rangle$ and $\mathcal{V}\langle\mathcal{U}'\rangle$. As a notation convention, we will use $\Pr[\cdot]$ and $\Pr'[\cdot]$ to denote the probability of some event in the execution of $\mathcal{A}(C, \mathcal{V}\langle\mathcal{U}\rangle)$ and $\mathcal{A}(C, \mathcal{V}\langle\mathcal{U}'\rangle)$, respectively. For any recommendation sequence $b = (b_1, \dotsc, b_T) \in \mathcal{B}_1 \times \dotsb \times \mathcal{B}_T$ and any round $t$, we define the random variables $\mathcal{E}_t(b) = \ln{\frac{\Pr[b_t | b_1, \dotsc, b_{t - 1}]}{\Pr'[b_t | b_1, \dotsc, b_{t - 1}]}}$ and $\mathcal{E}(b) = \ln{\frac{\Pr[b]}{\Pr'[b]}}$. It then follows that $\mathcal{E}(b) = \sum_{t = 1}^T{\mathcal{E}_t(b)}$. We also define random variable $\mathcal{L}_t$ to be the loss of execution $\mathcal{A}(C, \mathcal{V}\langle\mathcal{U}\rangle)$ in the $t$th round.

Finally, we list some useful properties of the following ``truncated-and-shifted exponential'' function
\begin{equation*}
\phi(x) = \begin{cases}
0 & \text{if $x \le \rho$},\\
e^{\lambda x} - e^{\lambda \rho} & \text{otherwise},
\end{cases}
\end{equation*}
where $-\infty < \rho < \infty$ and $\lambda > 0$. These properties will be useful when analyzing our algorithms.

\begin{lemma}\label{lemma:phi(x)-property-1}
$\phi(x) \le e^{\lambda |x - x'|}\phi(x') + (e^{\lambda |x - x'|} - 1)e^{\lambda \rho}$ for any reals $x, x'$.
\end{lemma}
\begin{proof}
If $x \le x'$, since $\phi(x)$ is non-decreasing, $\phi(x) \le \phi(x')$ and this lemma holds trivially. When $x > x'$, we have two cases:
\begin{itemize}
\item Case~1: $x' \le \rho$. In this case, $\phi(x') = 0$ and
\begin{align*}
\phi(x) & = \phi(x' + (x - x'))\\
& \le \phi(\rho + (x - x'))\\
& = (e^{\lambda (x - x')} - 1)e^{\lambda \rho}\\
& = e^{\lambda (x - x')}\phi(x') + (e^{\lambda (x - x')} - 1)e^{\lambda \rho};
\end{align*}
\item Case~2: $x' > \rho$. In this case, we have
\begin{equation*}
\frac{\phi(x) + e^{\lambda \rho}}{\phi(x') + e^{\lambda \rho}} = \frac{e^{\lambda x}}{e^{\lambda x'}} = e^{\lambda(x - x')}.
\end{equation*}
We get the desired result by rearranging terms.
\end{itemize}
This finishes the proof of Lemma~\ref{lemma:phi(x)-property-1}.
\end{proof}

\begin{lemma}\label{lemma:phi(x)-property-2}
$\frac{\phi(x)}{\phi(x')} \le e^{2\lambda (x - x')}$ for any $x \ge x' \ge \rho + \frac{\ln{2}}{\lambda}$.
\end{lemma}
\begin{proof}
Given $x \ge x' > \rho$, it follows that
\begin{align*}
& \quad\ \frac{e^{\lambda x} - e^{\lambda \rho}}{e^{\lambda x'} - e^{\lambda \rho}} \le e^{2\lambda (x - x')}\\
& \Leftrightarrow e^{\lambda x} - e^{\lambda \rho} \le e^{2\lambda (x - x')}(e^{\lambda x'} - e^{\lambda \rho})\\
& \Leftrightarrow e^{\lambda x'}(e^{2\lambda(x - x')} - e^{\lambda(x - x')}) \ge e^{\lambda \rho}(e^{2\lambda(x - x')} - 1)\\
& \Leftarrow e^{\lambda x'}e^{\lambda(x - x')} \ge e^{\lambda \rho}(e^{\lambda(x - x')} + 1)\\
& \Leftrightarrow \frac{e^{\lambda x'}}{e^{\lambda \rho}} \ge \frac{e^{\lambda(x - x')} + 1}{e^{\lambda(x - x')}}\\
& \Leftarrow \frac{e^{\lambda x'}}{e^{\lambda \rho}} \ge 2\\
& \Leftrightarrow x' \ge \rho + \frac{\ln{2}}{\lambda}.
\end{align*}
This finishes the proof of Lemma~\ref{lemma:phi(x)-property-2}.
\end{proof}

\begin{lemma}\label{lemma:phi(x)-property-3}
Let $\theta_1, \dotsc, \theta_n$ be non-negative reals. Then
\begin{equation*}
\sum_{i = 1}^{n}{\phi(x_i + \theta_i)} + \phi(x_{n + 1}) \le \sum_{i = 1}^{n}{\phi(x_i)} + \phi\left(x_{n + 1} + \sum_{i = 1}^{n}{\theta_i}\right)
\end{equation*}
for any reals $x_1 \le \dotsb \le x_n \le x_{n + 1}$.
\end{lemma}
\begin{proof}
We prove this lemma by induction. When $n = 1$, it is easy to verify that $\phi(x_1 + \theta_1) + \phi(x_2) \le \phi(x_1) + \phi(x_2 + \theta_1)$ if $x_1 \le x_2$. Suppose this lemma holds for all $n = 1, \dotsc, k$. When $n = k + 1$, we have
\begin{align*}
& \quad\ \sum_{i = 1}^{k + 1}{\phi(x_i + \theta_i)} + \phi(x_{k + 2})\\
& = \sum_{i = 1}^{k}{\phi(x_i + \theta_i)} + \phi(x_{k + 1} + \theta_{k + 1}) + \phi(x_{k + 2})\\
& \le \sum_{i = 1}^{k}{\phi(x_i)} + \phi(x_{k + 1} + \theta_{k + 1}) + \phi\left(x_{k + 2} + \sum_{i = 1}^{k}{\theta_i}\right) && \tag{by induction hypothesis}\\
& \le \sum_{i = 1}^{k + 1}{\phi(x_i)} + \phi\left(x_{k + 2} + \sum_{i = 1}^{k + 1}{\theta_i}\right). && \tag{by induction hypothesis, since $x_{k + 2} + \sum_{i = 1}^{k}{\theta_i} \ge x_{k + 1}$}
\end{align*}
This finishes the induction step.
\end{proof}

\section{A Lower Bound on the Worst-case Expected Loss}
Before considering privacy, let us first look at the best achievable loss. We have the following lower bound on the worst-case expected loss, even for algorithms providing no privacy guarantee:
\begin{theorem}\label{theorem:lower-bound-loss}
The worst-case expected loss of any recommendation algorithm is $\Omega(R + \log{\frac{n}{P}})$, even for an algorithm providing no privacy guarantee and an oblivious adversary.
\end{theorem}
\begin{proof}[Proof of Theorem~\ref{theorem:lower-bound-loss}]
In \cite{yu2009dsybil}, Yu et~al.\ have shown that the worst-case expected loss is $\Omega(\log{n})$ for the case of $P = 1$ (Theorem 6 in \cite{yu2009dsybil}). A direct adaption of their proof shows that the worst-case expected loss is at least $\frac{1}{2}(m - 1)\lfloor\frac{\log{(n / P)}}{\log{m}}\rfloor = \Omega(\log{\frac{n}{P}})$ for any $P \ge 1$. In the remaining parts of this proof, we show that the worst-case expected loss is $\Omega(R)$.
	
We will show this by constructing a client with random preferences. In each of the first $R$ rounds, we independently and uniformly at random choose one object. We let the client like this object, and we let he/she dislike the other objects. In each of the remaining $T - R$ rounds, there is no randomness, the client always likes the object with the smallest index in that round.
	
There is a fixed voting pattern $\mathcal{V}\langle\mathcal{U}\rangle$ containing $P$ voters. In each round, all the $P$ voters always vote on the object with the smallest index in that round. By doing so, we ensure that all these $P$ voters are always peers of the client regardless of the instantiation of the client's random preferences.
	
One easily sees that the expected loss is at least $(1 - \frac{1}{m})R$ for any recommendation algorithm $\mathcal{A}$ when running with the random client and the voting pattern $\mathcal{V}\langle\mathcal{U}\rangle$ (over both the client's coin flips and $\mathcal{A}$'s coin flips), because the expected loss in each of the first $R$ rounds is $1 - \frac{1}{m}$, no matter what $\mathcal{A}$ does. Therefore, there must exist one instantiation of preferences for the client such that the expected loss is at least $(1 - \frac{1}{m})R$. This finishes the proof.
\end{proof}

\section{The Special Setting where \texorpdfstring{$D = R = 0$}{D = R = 0}}
Now let us consider differential privacy in recommendation systems. In order to better explain our ideas, we start by discussing the simple (yet non-trivial) setting where $D = R = 0$. That is, the client likes exactly one object in every round, and the client's peers never vote on any object that the client dislikes. We discuss the general setting where $D + R \ge 0$ in the next section.

\subsection{Lower Bound}
When $D = R = 0$, we have the following lower bound on the best achievable privacy.
\begin{theorem}\label{theorem:lower-bound-privacy-simple-setting}
For any recommendation algorithm that guarantees $L = o(T)$ worst-case expected loss and preserves $\epsilon$-differential privacy, if $D = R = 0$, then $\epsilon = \Omega(\frac{1}{P}\log\frac{T}{L})$, even for an oblivious adversary.
\end{theorem}
Theorem~\ref{theorem:lower-bound-privacy-simple-setting} directly follows from the following Lemma~\ref{lemma:lower-bound-privacy-simple-setting}, by plugging $L$ worst-case expected loss into it:
\begin{lemma}\label{lemma:lower-bound-privacy-simple-setting}
If $D = R = 0$, then the worst-case expected loss is at least $\frac{1}{2}Te^{-P\epsilon}$ for any recommendation algorithm that preserves $\epsilon$-differential privacy, even for an oblivious adversary.
\end{lemma}
\begin{proof}
In this proof, we consider a particular setup in which there are two objects $\alpha_t$ and $\beta_t$ in each round $t$. Given a client $C$, let $\tilde{\mathcal{V}}_C$ denote the voting pattern which contains only $P$ voters who are all $C$'s peers. We will construct a client $C^*$ for any recommendation algorithm $\mathcal{A}$, such that the expected loss of $\mathcal{A}(C^*, \tilde{\mathcal{V}}_{C^*})$ is at least $\frac{1}{2}Te^{-P\epsilon}$.

First consider a particular round $t$, we have the following claim:
\begin{claim}\label{claim:lower-bound-privacy-simple-setting}
If two clients Alice and Bob have the same preferences in all but the $t$th round, then one of $\mathcal{A}(\text{Alice}, \tilde{\mathcal{V}}_\text{Alice})$ and $\mathcal{A}(\text{Bob}, \tilde{\mathcal{V}}_\text{Bob})$ has at least $\frac{1}{2}e^{-P\epsilon}$ expected loss in the $t$th round.
\end{claim}
To prove this claim, assume that Alice likes $\alpha_t$ but Bob likes $\beta_t$. We construct a chain of $2P + 1$ voting patterns:
\begin{equation*}
\mathcal{V}_A^{(1)}, \dotsc, \mathcal{V}_A^{(P)}, \mathcal{V}, \mathcal{V}_B^{(P)}, \dotsc, \mathcal{V}_B^{(1)}.
\end{equation*}
The voting in the $t$th round of these voting patterns is as follows: in the middle voting pattern $\mathcal{V}$, there are $P$ voters voting on $\alpha_t$ (who are Alice's peers) and another $P$ voters voting on $\beta_t$ (who are Bob's peers). To get voting pattern $\mathcal{V}_A^{(P)}$, we remove one of Bob's peers from $\mathcal{V}$. We keep on doing this until we get $\mathcal{V}_A^{(1)}$, which contains only Alice's $P$ peers. Similarly, we construct $\mathcal{V}_B^{(P)}, \dotsc, \mathcal{V}_B^{(1)}$ by repeatedly removing Alice's peers. The voting in all the other $T - 1$ rounds is simple: in every voting pattern, all the voters vote on the (common) object that is liked by both Alice and Bob in every round. The reader can verify that $\mathcal{V}_A^{(1)} = \tilde{\mathcal{V}}_{\text{Alice}}$ and $\mathcal{V}_B^{(1)} = \tilde{\mathcal{V}}_{\text{Bob}}$.

We will use the middle voting pattern $\mathcal{V}$ as a bridge to relate the execution $\mathcal{A}(\text{Alice}, \mathcal{V}_A^{(1)})$ with the execution $\mathcal{A}(\text{Bob}, \mathcal{V}_B^{(1)})$. Since Alice and Bob have the same preferences in the first $t - 1$ rounds, the executions $\mathcal{A}(\text{Alice}, \mathcal{V})$ and $\mathcal{A}(\text{Bob}, \mathcal{V})$ have the same distribution of recommending objects in the $t$th round. Denote the common probability that the object $\beta_t$ is recommended in the $t$th round by $p$.

Let $p_A$ be the probability that $\beta_t$ is recommended in the $t$th round of the execution $\mathcal{A}(\text{Alice}, \mathcal{V}_A^{(1)})$. Notice that $p_A$ is also the expected loss of $\mathcal{A}(\text{Alice}, \mathcal{V}_A^{(1)})$ in the $t$th round. Since $\mathcal{V}_A^{(1)}$ and $\mathcal{V}$ are $P$-step adjacent voting patterns, we have $p_A \ge p \cdot e^{-P\epsilon}$. Similarly, let $p_B$ be the probability that $\alpha_t$ is recommended in the $t$th round of the execution $\mathcal{A}(\text{Bob}, \mathcal{V}_B^{(1)})$, then $p_B$ is the expected loss of $\mathcal{A}(\text{Bob}, \mathcal{V}_B^{(1)})$ in the $t$th round and $p_B \ge (1 - p) \cdot e^{-P\epsilon}$. In summary, we always have $\max\{p_A, p_B\} \ge \max\{p, 1 - p\} \cdot e^{-P\epsilon} \ge \frac{1}{2}e^{-P\epsilon}$ and therefore at least one of Alice and Bob incurs at least $\frac{1}{2}e^{-P\epsilon}$ expected loss in the $t$th round. This proves our claim.

We now inductively construct the client $C^*$ who incurs $\frac{1}{2}e^{-P\epsilon}$ expected loss in every round when running with $\tilde{\mathcal{V}}_{C^*}$. By the above claim, there exists a client $C_1$ who incurs at least $\frac{1}{2}e^{-P\epsilon}$ expected loss in the first round when running with $\tilde{\mathcal{V}}_{C_1}$. This finishes the base case.

Assume $C_k$ is a client who incurs at least $\frac{1}{2}e^{-P\epsilon}$ expected loss in every one of the first $k$ rounds when running with $\tilde{\mathcal{V}}_{C_k}$. We construct another client $C_k'$ who has the same preferences with $C_k$ in all but the $(k + 1)$th round. Then the expected loss in every one of the first $k$ rounds is the same for both $\mathcal{A}(C_k, \tilde{\mathcal{V}}_{C_k})$ and $\mathcal{A}(C_k', \tilde{\mathcal{V}}_{C_k'})$ (because the voting in the first $k$ rounds is the same in both $\tilde{\mathcal{V}}_{C_k}$ and $\tilde{\mathcal{V}}_{C_k'}$), and by the above claim, one of them has at least $\frac{1}{2}e^{-P\epsilon}$ expected loss in the $(k + 1)$th round. In other words, one of $\mathcal{A}(C_k, \tilde{\mathcal{V}}_{C_k})$ and $\mathcal{A}(C_k', \tilde{\mathcal{V}}_{C_k'})$ has at least $\frac{1}{2}e^{-P\epsilon}$ expected loss in every one of the first $k + 1$ rounds. This finishes the induction step.
\end{proof}

\subsection{Algorithm}
We propose the following Algorithm~\ref{algorithm:p-REC-simple} for the simple setting where $D = R = 0$. As we will see, it is a special case of the general p-REC algorithm in Section~\ref{section:p-REC}. Therefore, we call it the p-REC$_\text{sim}$ algorithm (``sim'' is short for ``simple'').
\begin{algorithm}[h]
	\SetKwInOut{Input}{Input}
	\SetKwInOut{Output}{Output}
	\SetKwInOut{Initialization}{Initialization}
	\SetKwData{Weight}{weight}
	\SetKwData{Obj}{obj}
	\SetKwData{Feedback}{feedback}
	\SetKwFunction{Main}{Main}
	\SetKwFunction{RecommendByWeights}{RecommendByWeight}
	\SetKwFunction{UpdateWeights}{UpdateWeight}
	
	\Input{A client $C$, a voting pattern $\mathcal{V}\langle\mathcal{U}\rangle$}
	\Output{Recommend an object from $\mathcal{B}_t$ to client $C$ in each round $t$}
	\Initialization{$\gamma \gets \frac{m}{3T - 1}$, $\lambda \gets 2m\ln{T}$, $\rho \gets \frac{1}{2m}$, $\Weight[i] \gets 1$ for each $i \in \mathcal{U}$}
	
	\BlankLine
	
	\SetKwBlock{MainProcedure}{Procedure \Main{}}{end}
	\MainProcedure{
		\ForEach{round $t = 1, \dotsc, T$}{
			$\Obj \gets$ \RecommendByWeights{$\Weight[\;]$}\;
			Recommend object $\Obj$ to the client $C$\;
			$\Feedback \gets$ the client $C$'s feedback on object $\Obj$\;
			\UpdateWeights{$\Weight[\;]$, $\Obj$, $\Feedback$}\;
		}
	}
	
	\SetKwBlock{RecommendByWeightsProcedure}{Procedure \RecommendByWeights{$\Weight[\;]$}}{end}
	\RecommendByWeightsProcedure{
		\ForEach{object $j \in \mathcal{B}_t$}{
			$x_{j, t} \gets \frac{\sum_{i \in \mathcal{U}_{j, t}}{\Weight[i]}}{\sum_{i \in \mathcal{U}}{\Weight[i]}}$, where $\mathcal{U}_{j, t}$ is the set of voters who vote on object $j$ in round $t$\;
		}
		Independently draw a Bernoulli random variable $Z_t$ with $\Pr[Z_t = 1] = \gamma$\;
		\eIf{$Z_t = 1$}{
			Independently draw an object $\Obj$ from $\mathcal{B}_t$ uniformly at random\;
		}
		{
			Independently draw an object $\Obj$ from $\mathcal{B}_t$ according to the following distribution: each object $j \in \mathcal{B}_t$ is drawn with probability proportional to $\phi(x_{j, t})$, where
			$
			\phi(x) = \begin{cases}
			0 & \text{if $x \le \rho$,}\\
			e^{\lambda x} - e^{\lambda \rho} & \text{otherwise\;}
			\end{cases}
			$
		}
		\Return $\Obj$\;
	}

	\SetKwBlock{UpdateWeightsProcedure}{Procedure \UpdateWeights{$\Weight[\;]$, $\Obj$, $\Feedback$}}{end}
	\UpdateWeightsProcedure{
		\eIf{$\Feedback =$ ``dislike''}{
			$\Weight[i] \gets 0$ for every voter $i$ who votes on object $\Obj$\;
		}
		{
			$\Weight[i] \gets 0$ for every voter $i$ who does not vote on object $\Obj$\;
		}
	}

	\caption{The p-REC$_\text{sim}$ algorithm.}\label{algorithm:p-REC-simple}
\end{algorithm}

The p-REC$_\text{sim}$ algorithm maintains a weight value $\mathsf{weight}[i]$ for each voter $i$, and it recommends objects according to voters' weight in each round by invoking the procedure \texttt{RecommendByWeight()}. When it receives the client's feedback, it invokes the procedure \texttt{UpdateWeight()} to update voters' weight. In p-REC$_\text{sim}$, each voter's weight is either $1$ or $0$. A voter with $0$ weight has no impact on the algorithm's output, and once a voter's weight is set to $0$, it will never be reset to $1$. Therefore, we can think of that \texttt{UpdateWeight()} works by ``kicking out'' voters from the system. We call the voters who have not been kicked out (i.e., those who have non-zero weight) \emph{surviving voters}.

The p-REC$_\text{sim}$ algorithm shares a similar structure to the classic Weighted Average algorithm for the PEA problem~\cite{cesa2006prediction}, as they both introduce weight to voters and output according to the weight. Our core contribution is the dedicated probability of recommending objects. In each round $t$, p-REC$_\text{sim}$ recommends object $j$ with probability $\gamma \cdot \frac{1}{m} + (1 - \gamma) \cdot \frac{\phi(x_{j, t})}{\sum_{k \in \mathcal{B}_t}{\phi(x_{k, t})}}$, where $x_{j, t}$ is the fraction of surviving voters voting on object $j$ in round $t$. We have:
\begin{theorem}\label{theorem:p-REC-sim-loss-and-privacy}
If $D = R = 0$, then the p-REC$_\text{sim}$ algorithm guarantees $O(\log\frac{n}{P})$ worst-case expected loss and it preserves $O(\frac{1}{P}\log{T})$-differential privacy, even for an adaptive adversary.
\end{theorem}
According to Theorem~\ref{theorem:lower-bound-loss}, p-REC$_\text{sim}$'s loss is within a constant factor to the optimal. Then by Theorem~\ref{theorem:lower-bound-privacy-simple-setting}, p-REC$_\text{sim}$'s $\epsilon$ is also within a constant factor to the optimal $\epsilon$ among all the algorithms that guarantee $O(T^\beta)$ worst-case expected loss.

\paragraph{Loss analysis of the p-REC$_\text{sim}$ algorithm.} In the remaining parts of this section, we prove Theorem~\ref{theorem:p-REC-sim-loss-and-privacy}. First, we analyze p-REC$_\text{sim}$'s loss:
\begin{theorem}\label{theorem:upper-bound-loss-simple-setting}
If $D = R = 0$, then the p-REC$_\text{sim}$ algorithm guarantees $O(\log\frac{n}{P})$ worst-case expected loss, even for an adaptive adversary.
\end{theorem}
\begin{proof}
First consider all the rounds with $Z_t = 0$. In these rounds, p-REC$_\text{sim}$ kicks out at least $\rho$ fraction of surviving voters whenever the client incurs a loss, because only the objects voted by at least $\rho$ fraction of surviving voters can be recommended. Since the number of surviving voters is initially at most $n$ and always at least $P$, it follows that the number of loss rounds is at most $\log_{\frac{1}{1 - \rho}}{\frac{n}{P}} \le 2m\ln{\frac{n}{P}}$. For those rounds with $Z_t = 1$, since there are $\gamma T$ such rounds on expectation, they cause at most $\gamma T \le \frac{m}{2}$ additional expected loss.
\end{proof}
A high probability result can also be derived by applying Chernoff's bound on the random variables $Z_t$'s.

\paragraph{Privacy analysis of the p-REC$_\text{sim}$ algorithm.} Next, we analyze p-REC$_\text{sim}$'s privacy. We aim to show that
\begin{theorem}\label{theorem:upper-bound-privacy-simple-setting}
If $D = R = 0$, then the p-REC$_\text{sim}$ algorithm preserves $O(\frac{1}{P}\log{T})$-differential privacy, even for an adaptive adversary.
\end{theorem}
To show Theorem~\ref{theorem:upper-bound-privacy-simple-setting}, consider the executions p-REC$_\text{sim}(C, \mathcal{V}\langle\mathcal{U}\rangle)$ and p-REC$_\text{sim}(C, \mathcal{V}\langle\mathcal{U}'\rangle)$, where $C$ is any client and $\mathcal{V}\langle\mathcal{U}\rangle$ and $\mathcal{V}\langle\mathcal{U}'\rangle$ are any pair of adjacent voting patterns ($\mathcal{U}$ contains one more voter than $\mathcal{U}'$). To show that p-REC$_\text{sim}$ preserves $O(\frac{1}{P}\log{T})$-differential privacy, it is sufficient to show that $|\mathcal{E}(b)| = O(\frac{1}{P}\log{T})$ for any recommendation sequence $b = (b_1, \dotsc, b_T)$. From now on, we will consider a fixed $b$, a fixed $C$ and a fixed pair of $\mathcal{V}\langle\mathcal{U}\rangle$ and $\mathcal{V}\langle\mathcal{U}'\rangle$.

Given the recommendation sequence $b = (b_1, \dotsc, b_T)$, let
\begin{equation*}
W_t(b) = \sum_{i \in \mathcal{U}}\mathsf{weight}[i]
\end{equation*}
be the number of surviving voters at the beginning of round $t$ in execution p-REC$_\text{sim}(C, \mathcal{V}\langle\mathcal{U}\rangle)$, conditioned on that the recommendations in the first $t - 1$ rounds are $b_1, \dotsc, b_{t - 1}$. Since p-REC$_\text{sim}$ never kicks out the client's peers, $W_t(b) \ge P \ge 6m$.

First, we upper-bound the ``privacy leakage'' in each single round:
\begin{lemma}\label{lemma:single-round-privacy-leakage-simple-setting}
For any round $t$, $|\mathcal{E}_t(b)| \le 3\lambda \cdot \frac{1}{W_t(b)}$.
\end{lemma}
Lemma~\ref{lemma:single-round-privacy-leakage-simple-setting} upper-bounds the ``privacy leakage'' in any single round. This lemma would be easy to prove if we were to recommend object $j$ in round $t$ with probability $\propto \exp(\lambda x_{j, t})$, because if there were $W_t(b)$ surviving voters, then the values of $x_{j, t}$ could only differ by $\frac{1}{W_t(b)}$ in the executions of p-REC$_\text{sim}(C, \mathcal{V}\langle\mathcal{U}\rangle)$ and p-REC$_\text{sim}(C, \mathcal{V}\langle\mathcal{U}'\rangle)$. By some relatively straightforward (but rather tedious) calculations, it can be shown that Lemma~\ref{lemma:single-round-privacy-leakage-simple-setting} is also true if we recommend object $j$ in round $t$ with probability $\gamma \cdot \frac{1}{m} + (1 - \gamma) \cdot \frac{\phi(x_{j, t})}{\sum_{k \in \mathcal{B}_t}{\phi(x_{k, t})}}$. Roughly speaking, this is because that the distortion $\gamma$ is small, and $\phi(x) \approx e^{\lambda x}$.
\begin{proof}[Proof of Lemma~\ref{lemma:single-round-privacy-leakage-simple-setting}]
Recall that $\mathcal{U}$ contains one more voter than $\mathcal{U}'$. If the extra voter in $\mathcal{U}$ has zero weight in the $t$th round (i.e., the extra voter has been kicked out), then we have $\Pr[b_t | b_1, \dotsc, b_{t - 1}] = \Pr'[b_t | b_1, \dotsc, b_{t - 1}]$ and the proof is done. In the remaining parts of this proof, we assume that the extra voter has $1$ weight.

Consider the $t$th round. Index the $m$ candidate objects in $\mathcal{B}_t$ by $\{1, \dotsc, m\}$. We use $x_j$ to denote the fraction of surviving voters voting on the $j$th object in the execution p-REC$_\text{sim}(C, \mathcal{V}\langle\mathcal{U}\rangle)$, where $j = 1, \dotsc, m$. Similarly, we use $x_j'$ to denote the corresponding fraction in the execution p-REC$_\text{sim}(C, \mathcal{V}\langle\mathcal{U}'\rangle)$. Since there are $W_t(b)$ surviving voters in the execution p-REC$_\text{sim}(C, \mathcal{V}\langle\mathcal{U}\rangle)$, the fractions $x_j$ and $x_j'$ can differ by at most $\frac{1}{W_t(b)}$, i.e., $|x_j - x_j'| \le \frac{1}{W_t(b)}$ for any $j = 1, \dotsc, m$.

Suppose the index of the recommended object $b_t$ is $j^*$, according to the p-REC$_\text{sim}$ algorithm,
\begin{equation*}
\Pr[b_t | b_1, \dotsc, b_{t - 1}] = \gamma \cdot \frac{1}{m} + (1 - \gamma) \cdot \frac{\phi(x_{j^*})}{\sum_{j = 1}^{m}{\phi(x_j)}}
\end{equation*}
and
\begin{equation*}
\Pr{}'[b_t | b_1, \dotsc, b_{t - 1}] = \gamma \cdot \frac{1}{m} + (1 - \gamma) \cdot \frac{\phi(x_{j^*}')}{\sum_{j = 1}^{m}{\phi(x_j')}}.
\end{equation*}
For simplicity, let us write $W_t(b)$ as $W_t$ hereafter in this proof. Our goal is to show that $e^{-3\lambda / W_t} \le \frac{\Pr[b_t | b_1, \dotsc, b_{t - 1}]}{\Pr'[b_t | b_1, \dotsc, b_{t - 1}]} \le e^{3\lambda / W_t}$, which is equivalent to
\begin{equation}\label{equation:upper-bound-equation-1}
\begin{cases}
\frac{\phi(x_{j^*})}{\sum_{j = 1}^{m}{\phi(x_j)}} \le e^{3\lambda / W_t}\frac{\phi(x_{j^*}')}{\sum_{j = 1}^{m}{\phi(x_j')}} + \frac{(e^{3\lambda / W_t} - 1)\gamma}{m(1 - \gamma)},\\
\frac{\phi(x_{j^*}')}{\sum_{j = 1}^{m}{\phi(x_j')}} \le e^{3\lambda / W_t}\frac{\phi(x_{j^*})}{\sum_{j = 1}^{m}{\phi(x_j)}} + \frac{(e^{3\lambda / W_t} - 1)\gamma}{m(1 - \gamma)}.
\end{cases}
\end{equation}

Recall that $|x_{j^*} - x_{j^*}'| \le \frac{1}{W_t}$, by Lemma~\ref{lemma:phi(x)-property-1}, we have
\begin{equation*}
\phi(x_{j^*}) \le e^{\lambda / W_t}\phi(x_{j^*}') + (e^{\lambda / W_t} - 1)e^{\lambda \rho}.
\end{equation*}
We also have $\frac{\sum_{j = 1}^{m}{\phi(x_j)}}{\sum_{j = 1}^{m}{\phi(x_j')}} \ge e^{-2\lambda / W_t}$ (see below, Lemma~\ref{lemma:ratio-of-sum-of-phi}). It then follows that
\begin{align*}
\frac{\phi(x_{j^*})}{\sum_{j = 1}^{m}{\phi(x_j)}} & \le \frac{e^{\lambda / W_t}\phi(x_{j^*}') + (e^{\lambda / W_t} - 1)e^{\lambda \rho}}{\sum_{j = 1}^{m}{\phi(x_j)}}\\
& = \frac{e^{\lambda / W_t}\phi(x_{j^*}')}{\sum_{j = 1}^{m}{\phi(x_j)}} + \frac{(e^{\lambda / W_t} - 1)e^{\lambda \rho}}{\sum_{j = 1}^{m}{\phi(x_j)}}\\
& \le \frac{e^{\lambda / W_t}\phi(x_{j^*}')}{e^{-2\lambda / W_t}\sum_{j = 1}^{m}{\phi(x_j')}} + \frac{(e^{\lambda / W_t} - 1)e^{\lambda \rho}}{\sum_{j = 1}^{m}{\phi(x_j)}}\\
& \le e^{3\lambda / W_t}\frac{\phi(x_{j^*}')}{\sum_{j = 1}^{m}{\phi(x_j')}} + \frac{(e^{\lambda / W_t} - 1)e^{\lambda \rho}}{\phi(\frac{1}{m})}.
\end{align*}
Notice in the last step, we use the fact that $\sum_{j = 1}^{m}{\phi(x_j)} \ge \phi(\frac{1}{m})$ because there always exists one $j$ with $x_j \ge \frac{1}{m}$. One can verify that $\frac{(e^{\lambda / W_t} - 1)e^{\lambda \rho}}{\phi(\frac{1}{m})} \le \frac{(e^{3\lambda / W_t} - 1)\gamma}{m(1 - \gamma)}$ by our choice of $\gamma$ and $\lambda$, hence we have proved the first inequality of \eqref{equation:upper-bound-equation-1}. We can show the second inequality of \eqref{equation:upper-bound-equation-1} in a completely symmetric way and prove this lemma.
\end{proof}

In the proof of Lemma~\ref{lemma:single-round-privacy-leakage-simple-setting}, we use the following lemma:
\begin{lemma}\label{lemma:ratio-of-sum-of-phi}
$e^{-2\lambda / W_t} \le \frac{\sum_{j = 1}^{m}{\phi(x_j)}}{\sum_{j = 1}^{m}{\phi(x_j')}} \le e^{2\lambda / W_t}$ ($x_j$'s and $x_j'$'s are defined in the proof of Lemma~\ref{lemma:single-round-privacy-leakage-simple-setting}).
\end{lemma}
\begin{proof}
Suppose the extra voter votes on the object with index $k$. We then have $x_k > x_k'$, and $x_j < x_j'$ for all the other $j \neq k$. Since $\phi(x)$ is non-decreasing, it then follows that $\phi(x_k) \ge \phi(x_k')$, and $\phi(x_j) \le \phi(x_j')$ for all the other $j \neq k$.

Let $\theta_j = |x_j - x_j'|$, we have
\begin{equation*}
\sum_{j \neq k}{\theta_j} = \theta_k \le \frac{1}{W_t} \le \frac{1}{P} \le \frac{1}{6m}.
\end{equation*}
Let $\ell = \arg\max_j\{x_1', \dotsc, x_j', \dotsc, x_m'\}$, then $x_\ell' \ge \frac{1}{m}$ and $x_\ell \ge x_\ell' - \theta_l \ge x_\ell' - \frac{1}{6m} \ge \frac{5}{6} \cdot \frac{1}{m}$. Depending on whether the extra voter votes on the object with index $\ell$ or not, we have two cases:
\begin{itemize}
\item \textit{Case~1:} $k = \ell$. Partition all the $m$ objects into two disjoint subsets: $\{\ell\}$ and $A = \{1, \dotsc, m\} \setminus \{\ell\}$. We have that $\phi(x_\ell) \ge \phi(x_\ell')$ (because $\ell = k$) and $\sum_{j \in A}{\phi(x_j)} \le \sum_{j \in A}{\phi(x_j')}$ (because $\forall j \in A: \phi(x_j) \le \phi(x_j')$). We then have
\begin{align*}
\frac{\sum_{j = 1}^{m}{\phi(x_j)}}{\sum_{j = 1}^{m}{\phi(x_j')}} & = \frac{\phi(x_\ell) + \sum_{j \in A}{\phi(x_j)}}{\phi(x_\ell') + \sum_{j \in A}{\phi(x_j')}}\\
& \le \frac{\phi(x_\ell)}{\phi(x_\ell')}\notag\\
& \le e^{2\lambda / W_t}. && \tag{by Lemma~\ref{lemma:phi(x)-property-2}}
\end{align*}
For the other direction, it follows that
\begin{align*}
\frac{\sum_{j = 1}^{m}{\phi(x_j)}}{\sum_{j = 1}^{m}{\phi(x_j')}} & = \frac{\phi(x_\ell) + \sum_{j \in A}{\phi(x_j)}}{\phi(x_\ell') + \sum_{j \in A}{\phi(x_j')}}\\
& = \frac{\phi(x_\ell) + \sum_{j \in A}{\phi(x_j)}}{\phi(x_\ell') + \sum_{j \in A}{\phi(x_j + \theta_j)}}\\
& \ge \frac{\phi(x_\ell) + \sum_{j \in A}{\phi(x_j)}}{\phi(x_\ell' + \sum_{j \in A}{\theta_j}) + \sum_{j \in A}{\phi(x_j)}} && \tag{by Lemma~\ref{lemma:phi(x)-property-3}}\\
& = \frac{\sum_{j \in A}{\phi(x_j)} + \phi(x_\ell)}{\sum_{j \in A}{\phi(x_j)} + \phi(x_\ell)} && \tag{since $\sum_{j \in A}{\theta_j = \theta_\ell = x_\ell - x_\ell'}$}\\
& = 1.
\end{align*}

\item \textit{Case~2:} $k \neq \ell$. Partition all the $m$ objects into three disjoint subsets: $\{k\}$, $\{\ell\}$ and $ B = \{1, \dotsc, m\} \setminus \{k, \ell\}$. Since $\phi(x_j) \le \phi(x_j')$ for any $j \in B$, we have $\sum_{j \in B}{\phi(x_j)} \le \sum_{j \in B}{\phi(x_j')}$. It then follows that
\begin{align*}
\frac{\sum_{j = 1}^{m}{\phi(x_j)}}{\sum_{j = 1}^{m}{\phi(x_j')}} & = \frac{\phi(x_k) + \phi(x_\ell) + \sum_{j \in B}{\phi(x_j)}}{\phi(x_k') + \phi(x_\ell') + \sum_{j \in B}{\phi(x_j')}}\\
& \le \begin{cases}
1 & \text{if $\frac{\phi(x_k) + \phi(x_\ell)}{\phi(x_k') + \phi(x_\ell')} \le 1$},\\
\frac{\phi(x_k) + \phi(x_\ell)}{\phi(x_k') + \phi(x_\ell')} & \text{if $\frac{\phi(x_k) + \phi(x_\ell)}{\phi(x_k') + \phi(x_\ell')} > 1$}.
\end{cases}
\end{align*}
Thus, to upper-bound $\frac{\sum_{j = 1}^{m}{\phi(x_j)}}{\sum_{j = 1}^{m}{\phi(x_j')}}$, it is sufficient to show an upper bound for $\frac{\phi(x_k) + \phi(x_\ell)}{\phi(x_k') + \phi(x_\ell')}$. There are two subcases to consider:
\begin{itemize}
\item $x_k' \le x_\ell$. By Lemma~\ref{lemma:phi(x)-property-3}, it follows that
\begin{align}\label{equation:upper-bound-equation-2}
\frac{\phi(x_k) + \phi(x_\ell)}{\phi(x_k') + \phi(x_\ell')} & = \frac{\phi(x_k' + \theta_k) + \phi(x_\ell)}{\phi(x_k') + \phi(x_\ell')}\notag\\
& \le \frac{\phi(x_k') + \phi(x_\ell + \theta_k)}{\phi(x_k') + \phi(x_\ell')}.
\end{align}
Notice that $0 \le (x_\ell + \theta_k) - x_\ell' = \sum_{j \in B}{\theta_j} \le \frac{1}{W_t}$, and $x_\ell' \ge \rho + \frac{\ln{2}}{\lambda}$, according to Lemma~\ref{lemma:phi(x)-property-2}, we have
\begin{equation*}
1 \le \frac{\phi(x_\ell + \theta_k)}{\phi(x_\ell')} \le e^{2\lambda((x_\ell + \theta_k) - x_\ell')} \le e^{2\lambda / W_t}.
\end{equation*}
Substitute this to \eqref{equation:upper-bound-equation-2}, we get $\frac{\sum_{j = 1}^{m}{\phi(x_j)}}{\sum_{j = 1}^{m}{\phi(x_j')}} \le e^{2\lambda / W_t}$.

\item $x_k' > x_\ell$. In this case, since $\phi(x_k) \ge \phi(x_k') \ge \phi(x_\ell)$, $\phi(x_k') \le \phi(x_\ell')$, we have
\begin{equation}\label{equation:upper-bound-equation-3}
\frac{\phi(x_k) + \phi(x_\ell)}{\phi(x_k') + \phi(x_\ell')} \le \frac{\phi(x_k) + \phi(x_k)}{\phi(x_k') + \phi(x_k')} = \frac{\phi(x_k)}{\phi(x_k')}.
\end{equation}
Notice that $x_k' > x_\ell \ge \frac{5}{6} \cdot \frac{1}{m} \ge \rho + \frac{\ln{2}}{\lambda}$, by Lemma~\ref{lemma:phi(x)-property-2}, we have
\begin{equation*}
\frac{\phi(x_k)}{\phi(x_k')} \le e^{2\lambda(x_k - x_k')} \le e^{2\lambda / W_t}.
\end{equation*}
Substitute this to \eqref{equation:upper-bound-equation-3}, we get $\frac{\sum_{j = 1}^{m}{\phi(x_j)}}{\sum_{j = 1}^{m}{\phi(x_j')}} \le e^{2\lambda / W_t}$.
\end{itemize}

Now we have shown one direction of the desired inequality in \textit{Case~2}, for the other direction, it follows that
\begin{align}
\frac{\sum_{j = 1}^{m}{\phi(x_j)}}{\sum_{j = 1}^{m}{\phi(x_j')}} & = \frac{\phi(x_k) + \phi(x_\ell) + \sum_{j \in B}{\phi(x_j)}}{\phi(x_k') + \phi(x_\ell') + \sum_{j \in B}{\phi(x_j')}}\notag\\
& \ge \frac{\phi(x_\ell) + \sum_{j \in B}{\phi(x_j)}}{\phi(x_\ell') + \sum_{j \in B}{\phi(x_j')}} && \tag{since $\phi(x_k) \ge \phi(x_k')$}\notag\\
& = \frac{\phi(x_\ell) + \sum_{j \in B}{\phi(x_j)}}{\phi(x_\ell') + \sum_{j \in B}{\phi(x_j + \theta_j)}}\notag\\
& \ge \frac{\phi(x_\ell) + \sum_{j \in B}{\phi(x_j)}}{\phi(x_\ell' + \sum_{j \in B}{\theta_j}) + \sum_{j \in B}{\phi(x_j)}}.\label{equation:upper-bound-equation-4}
\end{align}
The last step is due to Lemma~\ref{lemma:phi(x)-property-3}. Notice that $0 \le (x_\ell' + \sum_{j \in B}{\theta_j}) - x_\ell = \sum_{j \in B \cup \{\ell\}}{\theta_j} = \theta_k \le \frac{1}{W_t}$, and $x_\ell' \ge \rho + \frac{\ln{2}}{\lambda}$, according to Lemma~\ref{lemma:phi(x)-property-2}, we have
\begin{equation*}
1 \le \frac{\phi(x_\ell' + \sum_{j \in B}{\theta_j})}{\phi(x_\ell)} \le e^{2\lambda / W_t}.
\end{equation*}
Substitute this to \eqref{equation:upper-bound-equation-4}, we get $\frac{\sum_{j = 1}^{m}{\phi(x_j)}}{\sum_{j = 1}^{m}{\phi(x_j')}} \ge e^{-2\lambda / W_t}$.
\end{itemize}
The proof of Lemma~\ref{lemma:ratio-of-sum-of-phi} is done because we have shown $e^{-2\lambda / W_t} \le \frac{\sum_{j = 1}^{m}{\phi(x_j)}}{\sum_{j = 1}^{m}{\phi(x_j')}} \le e^{2\lambda / W_t}$ in both \textit{Case~1} and \textit{Case~2}.
\end{proof}

Next, we show that a constant fraction of surviving voters are kicked out whenever there is non-zero ``privacy leakage:''
\begin{lemma}\label{lemma:round-with-non-zero-privacy-leakage-simple-setting}	For any round $t$, if $|\mathcal{E}_t(b)| \neq 0$, then $W_{t + 1}(b) \le W_t(b) \cdot (1 - \frac{1}{3m})$.
\end{lemma}
\begin{proof}
Notice that $|\mathcal{E}_t(b)| \neq 0$ iff $\Pr[b_t | b_1, \dotsc, b_{t - 1}] \neq \Pr{}'[b_t | b_1, \dotsc, b_{t - 1}]$. Let $x$ and $x'$ be the fraction of surviving voters voting on the recommended object $b_t$ in execution p-REC$_\text{sim}(C, \mathcal{V}\langle\mathcal{U}\rangle)$ and p-REC$_\text{sim}(C, \mathcal{V}\langle\mathcal{U}'\rangle)$, respectively. Since there are $W_t(b)$ surviving voters, $|x - x'| \le \frac{1}{W_t(b)} \le \frac{1}{P} \le \frac{1}{6m}$.

We claim that $x > \frac{1}{3m}$. Assume for contradiction that $x \le \frac{1}{3m}$. Since $|x - x'| \le \frac{1}{6m}$, both $x$ and $x'$ will be no larger than $\frac{1}{3m} + \frac{1}{6m} = \frac{1}{2m} = \rho$. Notice that $\phi(\zeta) = 0$ for any variable $\zeta \le \rho$, it then follows that $\phi(x) = \phi(x') = 0$ and $\Pr[b_t | b_1, \dotsc, b_{t - 1}] = \Pr'[b_t | b_1, \dotsc, b_{t - 1}] = \gamma \cdot \frac{1}{m}$, contradiction.

If the clients dislikes the recommended object $b_t$, by p-REC$_\text{sim}$'s rule of updating weight, $x > \frac{1}{3m}$ fraction of surviving voters will be kicked out.

If the clients likes $b_t$, then there must exist another object $\xi \in \mathcal{B}_t$ which is different from $b_t$, such that $\Pr[\xi | b_1, \dotsc, b_{t - 1}] \neq \Pr'[\xi | b_1, \dotsc, b_{t - 1}]$. Otherwise, if all the other objects are recommended with the same probability in executions p-REC$_\text{sim}(C, \mathcal{V}\langle\mathcal{U}\rangle)$ and p-REC$_\text{sim}(C, \mathcal{V}\langle\mathcal{U}'\rangle)$, so will be $b_t$, contradiction. By similar arguments, there are at least $\frac{1}{3m}$ fraction of surviving voters voting on the object $\xi$ in both p-REC$_\text{sim}(C, \mathcal{V}\langle\mathcal{U}\rangle)$ and p-REC$_\text{sim}(C, \mathcal{V}\langle\mathcal{U}'\rangle)$. Since p-REC$_\text{sim}$ kicks out all the voters who do not vote on $b_t$ (including those who vote on $\xi$), again we get the desired result.
\end{proof}

Lemma~\ref{lemma:single-round-privacy-leakage-simple-setting} states that $|\mathcal{E}_t(b)|$ is $O(\frac{\lambda}{W_t(b)}) = O(\frac{1}{P}\log{T})$. Lemma~\ref{lemma:round-with-non-zero-privacy-leakage-simple-setting} implies that there can be at most $O(\log{\frac{n}{P}})$ rounds with $|\mathcal{E}_t(b)| \neq 0$. A combination of these two lemmas immediately shows that overall we have $O(\frac{1}{P}\log{T} \cdot \log{\frac{n}{P}})$-differential privacy. With a bit more careful analysis, we can remove the extra $\log{\frac{n}{P}}$ factor and prove Theorem~\ref{theorem:upper-bound-privacy-simple-setting}:
\begin{proof}[Proof of Theorem~\ref{theorem:upper-bound-privacy-simple-setting}]
Let $t_1, \dotsc, t_K$ be the rounds in which $|\mathcal{E}_{t_i}(b)| \neq 0$, we have
\begin{equation*}
|\mathcal{E}(b)| = \left|\sum_{t = 1}^T{\mathcal{E}_t(b)}\right| \le \sum_{t = 1}^T{|\mathcal{E}_t(b)|} = \sum_{i = 1}^{K}{|\mathcal{E}_{t_i}(b)|} \le 3\lambda\sum_{i = 1}^{K}{\frac{1}{W_{t_i}(b)}}.
\end{equation*}
By Lemma~\ref{lemma:round-with-non-zero-privacy-leakage-simple-setting}, and the fact that $W_t(b)$ is non-increasing with respect to $t$, we have $\frac{1}{W_{t_i}(b)} \le \frac{1}{W_{t_{i + 1}}(b)} \cdot \left(1 - \frac{1}{3m}\right)$. Therefore the sequence $\{\frac{1}{W_{t_K}(b)}, \dotsc, \frac{1}{W_{t_1}(b)}\}$ is upper-bounded by the geometric sequence with $\frac{1}{W_{t_K}(b)}$ as the first term and $1 - \frac{1}{3m}$ as the common ratio. This implies that
\begin{equation*}
\sum_{i = 1}^{K}{\frac{1}{W_{t_i}(b)}} \le 3m \cdot \frac{1}{W_{t_K}(b)} \le 3m \cdot \frac{1}{P},
\end{equation*}
and
\begin{equation*}
|\mathcal{E}(b)| \le 3\lambda\sum_{i = 1}^{K}{\frac{1}{W_{t_i}(b)}} \le 18m^2 \cdot \frac{1}{P}\log{T} = O\left(\frac{1}{P}\log{T}\right),
\end{equation*}
as desired.
\end{proof}

\section{The General Setting where \texorpdfstring{$D + R \ge 0$}{D + R >= 0}}
\subsection{Lower Bound}\label{section:lower-bound-general-setting}
In this section, we prove Theorem~\ref{theorem:lower-bound-privacy}. If $0 \le D + R < 6\ln{T}$ and the target loss $L = O(T^\beta)$, then $\Omega(\log\frac{T}{L}) = \Omega(D + R + \log\frac{T}{L})$ and hence Theorem~\ref{theorem:lower-bound-privacy} is implied by Theorem~\ref{theorem:lower-bound-privacy-simple-setting}. When $D + R \ge 6\ln{T}$, we have the following Theorem~\ref{theorem:lower-bound-privacy-general-setting}. Theorem~\ref{theorem:lower-bound-privacy} is then proved because $\Omega(D + R) = \Omega(D + R + \log\frac{T}{L})$ if $D + R \ge 6\ln{T}$.
\begin{theorem}\label{theorem:lower-bound-privacy-general-setting}
For any recommendation algorithm that guarantees $L = o(T)$ worst-case expected loss and preserves $\epsilon$-differential privacy, if $D + R \ge 6\ln{T}$, then $\epsilon = \Omega(\frac{1}{P}(D + R))$, even for an oblivious adversary.
\end{theorem}

Before proving Theorem~\ref{theorem:lower-bound-privacy-general-setting}, we first explain the intuition behind the proof by a simple illustrative example. Imagine that there is one client Alice, and two voting patterns $\mathcal{V}_1$ and $\mathcal{V}_2$. Both $\mathcal{V}_1$ and $\mathcal{V}_2$ contain only one voter named Bob, but Bob may vote differently in $\mathcal{V}_1$ and $\mathcal{V}_2$. We let Bob be Alice's peer in both $\mathcal{V}_1$ and $\mathcal{V}_2$. For simplicity let us set $R = 0$, so Bob never votes on any object that Alice dislikes. By Definition~\ref{definition:adjacent-voting-pattern}, $\mathcal{V}_1$ and $\mathcal{V}_2$ are $2$-step voting patterns.

Now consider a particular round $t$ with two candidate objects. If Alice likes only one of the objects, then there is only one way for Bob to cast vote; otherwise Bob will no longer be a peer of Alice. However, if Alice likes both objects, then Bob can vote on different objects in $\mathcal{V}_1$ and $\mathcal{V}_2$ without breaking the promise that he is Alice's peer. Since Bob is the only information source of the system, an recommendation algorithm $\mathcal{A}$ has to somehow ``follow'' Bob, and hence the distributions of the executions $\mathcal{A}(\text{Alice}, \mathcal{V}_1)$ and $\mathcal{A}(\text{Alice}, \mathcal{V}_2)$ will be different. If Alice's diversity of preferences is $D$, then this situation can happen for $D$ times, which results an $\epsilon \propto D$. The linear dependency of $\epsilon$ on $R$ is for a similar reason.

\begin{proof}[Proof of Theorem~\ref{theorem:lower-bound-privacy-general-setting}]
We first prove Theorem~\ref{theorem:lower-bound-privacy-general-setting} for the case where $P = 1$ and $6\ln{T} \le D + R \le T$. Once we solve this case, the other cases can be easily solved and we discuss them in the end of this proof.
	
To show Theorem~\ref{theorem:lower-bound-privacy-general-setting} for the case where $P = 1$, it is sufficient to show that for any given algorithm $\mathcal{A}$, we can construct a client $C$ and a pair of $2$-step adjacent voting patterns $\mathcal{V}\langle\mathcal{U}\rangle, \mathcal{V}\langle\mathcal{U}'\rangle$, such that $\ln{\frac{\Pr[b]}{\Pr'[b]}} = \Omega(D + R)$ for some recommendation sequence $b \in \mathcal{B}_1 \times \dotsb \times \mathcal{B}_T$ and a sufficiently large $T$.
	
We consider a particular setup in which there are two objects $\alpha_t$ and $\beta_t$ in each round $t$. We construct the client $C$ by setting $C$'s preferences on $\alpha_t$'s and $\beta_t$'s. We will always ensure that $C$ likes both $\alpha_t$ and $\beta_t$ in at most $D$ rounds, hence that $D_C \le D$. For the voting patterns $\mathcal{V}\langle\mathcal{U}\rangle$ and $\mathcal{V}\langle\mathcal{U}'\rangle$, we let each of them contain one voter. Let $U$ and $U'$ be the only voter in voting pattern $\mathcal{V}\langle\mathcal{U}\rangle$ and $\mathcal{V}\langle\mathcal{U}'\rangle$, respectively. We ensure that both $U$ and $U'$ vote on at most $R$ objects that are disliked by the client $C$, hence both $U$ and $U'$ are $C$'s peers.
	
In sake of convenience, for any client $C$ and any pair of voting patterns $\mathcal{V}\langle\mathcal{U}\rangle$ and $\mathcal{V}\langle\mathcal{U}'\rangle$ as described above, we call the triple $(C, \mathcal{V}\langle\mathcal{U}\rangle, \mathcal{V}\langle\mathcal{U}'\rangle)$ a \emph{configuration}. We will introduce a random procedure of generating configurations for any given algorithm $\mathcal{A}$, and then use it to demonstrate the existence of one ``bad'' configuration that has large privacy leakage.
	
We construct a ``bad'' configuration round by round. Imagine that we are running the two executions $\mathcal{A}(C, \mathcal{V}\langle\mathcal{U}\rangle)$ and $\mathcal{A}(C, \mathcal{V}\langle\mathcal{U}'\rangle)$ simultaneously, and now $\mathcal{A}$ is at the beginning of the $t$th round, with the previous recommendation history being $(b_1, \dotsc, b_{t - 1})$. Let us write $b_{<t} = (b_1, \dotsc, b_{t - 1})$ for short. At this point, $\mathcal{A}$ will calculate the probabilities of recommending $\alpha_t$ and $\beta_t$, based on the voter $U$'s (or $U'$'s) vote.
	
The voter $U$ may vote on either $\alpha_t$ or $\beta_t$, and $\mathcal{A}$'s probability of recommending objects depends on how $U$ vote. We use $p_\alpha$ to denote the probability that the algorithm $\mathcal{A}$ ``follows'' the voter $U$ when $U$ votes on $\alpha_t$. That is, $p_\alpha$ is $\mathcal{A}$'s probability of recommending $\alpha_t$ when $U$ votes on $\alpha_t$, conditioned on the recommendation history is $b_{<t}$. Similarly, we define $p_\beta$ to be the probability that $\mathcal{A}$ ``follows'' the voter $U$ when $U$ votes on $\beta_t$. We define $p_\alpha'$ and $p_\beta'$ to be the counterparts of $p_\alpha$ and $p_\beta$ with voter $U$ replaced by voter $U'$. I.e., $p_\alpha'$($p_\beta'$) is the probability that $\mathcal{A}$ ``follows'' the voter $U'$ when $U'$ votes on $\alpha_t$($\beta_t$).
	
In order to better explain the intuition behind this proof, let us temporarily assume an adaptive adversary, therefore the adversary could see the recommendation history $b_{<t}$ and calculate $p_\alpha$, $p_\beta$, $p_\alpha'$ and $p_\beta'$ by ``simulating'' $\mathcal{A}$. It can then (adaptively) construct a ``bad'' configuration by choosing one of the following $4$ settings for round $t$, based on the values of $p_\alpha$, $p_\beta$, $p_\alpha'$ and $p_\beta'$:
\begin{itemize}
\item \textit{Setting~1}: the client likes $\alpha_t$, dislikes $\beta_t$, and both $U$ and $U'$ vote on $\alpha_t$;
\item \textit{Setting~2}: the client dislikes $\alpha_t$, likes $\beta_t$, and both $U$ and $U'$ vote on $\beta_t$;
\item \textit{Setting~3}: the client likes both $\alpha_t$ and $\beta_t$, $U$ votes on $\alpha_t$, and $U'$ votes on $\beta_t$;
\item \textit{Setting~4}: the client likes $\alpha_t$, dislikes $\beta_t$, $U$ votes on $\alpha_t$, and $U'$ votes on $\beta_t$.
\end{itemize}
The adaptive construction is as follows:
\begin{itemize}
\item \textit{Case~1}: $0 \le \min\{p_\alpha, p_\beta\} \le \frac{3}{4}$. We have two more subcases:
\begin{itemize}
\item \textit{Case~1.a}: If $p_\alpha \le \frac{3}{4}$, then the adversary chooses \textit{Setting~1}, in which case we have
\begin{align*}
\mathbb{E}[\mathcal{L}_t | b_{<t}] & = \Pr[\beta_t | b_{<t}]\\
& = 1 - p_\alpha\\
& \ge \frac{1}{4}.
\end{align*}
\item \textit{Case~1.b}: Else we must have $p_\beta \le \frac{3}{4}$, in which case the adversary chooses \textit{Setting~2}. We have
\begin{align*}
\mathbb{E}[\mathcal{L}_t | b_{<t}] & = \Pr[\alpha_t | b_{<t}]\\
& = 1 - p_\beta\\
& \ge \frac{1}{4}.
\end{align*}
\end{itemize}
In summary, in \textit{Case~1}, the adversary can always force the algorithm to have a constant expected loss in round $t$.
\item \textit{Case~2}: $\frac{3}{4} < \min\{p_\alpha, p_\beta\} \le 1$.In this case, algorithm $\mathcal{A}$ ``follows'' voter $U$ with a large probability (at least $\frac{3}{4}$) regardless how $U$ votes. We have two more subcases:
\begin{itemize}
\item \textit{Case~2.a}: $0 \le p_\beta' \le \frac{1}{2}$, in which case $\mathcal{A}$ ``follows'' voter $U'$ with only a small probability (at most $\frac{1}{2}$) if $U'$ votes on $\beta_t$. Recall that $\mathcal{A}$ ``follows'' voter $U$ with probability at least $\frac{3}{4}$ no matter which object $U$ votes on. Thus, if the adversary lets both $U$ and $U'$ vote on $\beta_t$, the two distributions $\Pr[\cdot | b_{<t}]$ and $\Pr'[\cdot | b_{<t}]$ will be sufficiently different, and the algorithm $\mathcal{A}$ will ``leak privacy.'' Specifically, the adversary will choose \textit{Setting~2} for \textit{Case~2.a} and we have
\begin{align*}
\mathbb{E}[\mathcal{E}_t | b_{<t}] & = \Pr[\alpha_t | b_{<t}] \cdot \ln{\frac{\Pr[\alpha_t | b_{<t}]}{\Pr'[\alpha_t | b_{<t}]}} + \Pr[\beta_t | b_{<t}] \cdot \ln{\frac{\Pr[\beta_t | b_{<t}]}{\Pr'[\beta_t | b_{<t}]}}\\
& = (1 - p_\beta)\ln{\frac{1 - p_\beta}{1 - p_\beta'}} + p_\beta\ln{\frac{p_\beta}{p_\beta'}}\\
& \ge \frac{1}{4}\ln{\frac{1 / 4}{1 / 2}} + \frac{3}{4}\ln{\frac{3 / 4}{1 / 2}} \tag{$\because p_\beta \ge \frac{3}{4}$ and $p_\beta' \le \frac{1}{2}$}\\
& > 0.13.
\end{align*}
That is, the expected ``privacy leakage'' in the $t$th round is at least a constant.
\item \textit{Case~2.b}: $\frac{1}{2} < p_\beta' \le 1$, in which case algorithm $\mathcal{A}$ ``follows'' voter $U'$ with a large probability (at least $\frac{1}{2}$) too if $U'$ votes on $\beta_t$. Since $\mathcal{A}$ ``follows'' voter $U$ with probability at least $\frac{3}{4}$ no matter which object $U$ votes on, to make the distributions $\Pr[\cdot | b_{<t}]$ and $\Pr'[\cdot | b_{<t}]$ be sufficiently different and force the algorithm to ``leak privacy,'' the adversary can let $U$ vote on $\alpha_t$ and $U'$ vote on $\beta_t$ --- this is how \textit{Setting~3} and \textit{Setting~4} are constructed. Notice that the adversary cannot keep choosing \textit{Setting~3} or \textit{Setting~4}, because the client likes more than one objects in at most $D$ rounds, and $U$ and $U'$ can vote on at most $R$ objects that are disliked by the client. In our construction, the adversary will choose \textit{Setting~3} for the first $D$ rounds that are in \textit{Case~2.b}, and then choose \textit{Setting~4} for the next $R$ rounds that are in \textit{Case~2.b}. The adversary arbitrarily chooses \textit{Setting~1} or \textit{Setting~2} for the remaining rounds that are in \textit{Case~2.b}.
			
If \textit{Setting~3} or \textit{Setting~4} is chosen, we have
\begin{align*}
\mathbb{E}[\mathcal{E}_t | b_{<t}] & = \Pr[\alpha_t | b_{<t}] \cdot \ln{\frac{\Pr[\alpha_t | b_{<t}]}{\Pr'[\alpha_t | b_{<t}]}} + \Pr[\beta_t | b_{<t}] \cdot \ln{\frac{\Pr[\beta_t | b_{<t}]}{\Pr'[\beta_t | b_{<t}]}}\\
& = p_\alpha\ln{\frac{p_\alpha}{1 - p_\beta'}} + (1 - p_\alpha)\ln{\frac{1 - p_\alpha}{p_\beta'}}\\
& \ge \frac{3}{4}\ln{\frac{3 / 4}{1 / 2}} + \frac{1}{4}\ln{\frac{1 / 4}{1 / 2}} \tag{$\because p_\alpha \ge \frac{3}{4}$ and $p_\beta' > \frac{1}{2}$}\\
& > 0.13.
\end{align*}
If \textit{Setting~1} or \textit{Setting~2} is chosen, we have
\begin{equation*}
\mathbb{E}[\mathcal{E}_t | b_{<t}] \ge 0
\end{equation*}
because the relative entropy is always non-negative.
\end{itemize}
\end{itemize}
	
It can be shown that with this adaptive construction, $\mathbb{E}[\mathcal{E}] = \Omega(D + R)$ for any algorithm $\mathcal{A}$, which implies the existence of one recommendation sequence $b$ such that $\mathcal{E}(b) = \ln{\frac{\Pr[b]}{\Pr'[b]}} = \Omega(D + R)$. To see why $\mathbb{E}[\mathcal{E}] = \Omega(D + R)$, we first notice that there cannot be too many rounds in \textit{Case~1} (including \textit{Case~1.a} and \textit{Case~1.b}) on expectation, because $\mathcal{A}$ has to ensure $o(T)$ worst-case expected loss. Therefore most of the rounds must be in \textit{Case~2.a} or \textit{Case~2.b}. If there are many rounds in \textit{Case~2.a}, then $\mathbb{E}[\mathcal{E}]$ must be large because $\mathbb{E}[\mathcal{E}_t | b_{<t}]$ is at least a constant in every \textit{Case~2.a} round, and $\mathbb{E}[\mathcal{E}_t | b_{<t}] \ge 0$ in all the other rounds (the relative entropy is always non-negative). Otherwise, there must be many rounds in \textit{Case~2.b}. In this case, the adversary can choose \textit{Setting~3} and \textit{Setting~4} for $\Omega(D + R)$ times, and we will have $\mathbb{E}[\mathcal{E}] = \Omega(D + R)$.
	
The drawback of the above arguments is of course the assumption of an adaptive adversary. To remove this assumption, an (oblivious) adversary will choose \textit{Setting~1}, \textit{Setting~2}, \textit{Setting~3} or \textit{Setting~4} randomly in each round, without looking at the recommendation history. In other words, we are to construct a random configuration.
	
The random construction is as follows. In each round $t$, the adversary independently draws a random variable $X_t$ such that
\begin{equation*}
X_t = \begin{dcases}
0 & \text{with probablitiy $\frac{D + R}{2T}$},\\
1 & \text{with probablitiy $\frac{1}{2}(1 - \frac{D + R}{2T})$},\\
2 & \text{with probablitiy $\frac{1}{2}(1 - \frac{D + R}{2T})$}.\\
\end{dcases}
\end{equation*}
The adversary chooses \textit{Setting~1} for the $t$th round if $X_t = 1$, and it chooses \textit{Setting~2} for the $t$th round round if $X_t = 2$. If $X_t = 0$, the adversary does the following: if it has chosen \textit{Setting~3} for less than $D$ times before the $t$th round, chooses \textit{Setting~3}; else chooses \textit{Setting~4}. Notice that in this random construction, it is possible to get an ``illegal'' configuration (i.e., a configuration with more than $R$ rounds in \textit{Setting~4}). However, since $\Pr[X_t = 0] = \frac{D + R}{2T}$, the expected number of rounds in which $X_t = 0$ is $\frac{1}{2}(D + R)$. By Chernoff's bound, the probability that we have more than $D + R$ rounds with $X_t = 0$ is at most $\frac{1}{T}$ (here the condition $D + R \ge 6\ln{T}$ is used). In other words, our construction generates a ``legal'' configuration with high probability.
	
This oblivious construction is analogous to the aforementioned adaptive one in the following sense: previously, if $\mathcal{A}$ is in \textit{Case~1.a}, \textit{Case~1.b} or \textit{Case~2.a}, it will encounter a ``bad'' setting (which is \textit{Setting~1}, \textit{Setting~2} and \textit{Setting~2}, respectively) \emph{for sure}; now $\mathcal{A}$ encounters a ``bad'' setting with a constant probability $\frac{1}{2}(1 - \frac{D + R}{2T}) \ge \frac{1}{4}$. Moreover, in the oblivious construction, we also ensure that there are sufficiently many rounds ($\frac{1}{2}(D + R)$ rounds on expectation) in which \textit{Setting~3} or \textit{Setting~4} is chosen.
	
Due to this analogousness between these two constructions, we can derive a similar lower bound for $\mathbb{E}[\mathcal{E}]$ and prove Theorem~\ref{theorem:lower-bound-privacy-general-setting} for the case of $P = 1$ and $6\ln{T} \le D + R \le T$:
	
Let $\mathcal{F}$ be the universe of all possible configurations, for any configuration $F = (C, \mathcal{V}\langle\mathcal{U}\rangle, \mathcal{V}\langle\mathcal{U}'\rangle) \in \mathcal{F}$, we use $\Pr[\cdot | F]$, $\Pr'[\cdot | F]$ and $\mathbb{E}[\cdot | F]$ to denote the corresponding probabilities and expectations in the executions of $\mathcal{A}(C, \mathcal{V}\langle\mathcal{U}\rangle)$ and $\mathcal{A}(C, \mathcal{V}\langle\mathcal{U}'\rangle)$.
	
Given any subset $\mathcal{G} \subseteq \mathcal{F}$, we define the following ``conditional'' relative entropies over $\mathcal{G}$:
\begin{equation*}
\mathbb{E}[\mathcal{E}_t | \mathcal{G}] = \sum_{F \in \mathcal{G}}{\mathbb{E}[\mathcal{E}_t | F] \cdot \Pr[F | \mathcal{G}]}\\
\end{equation*}
and
\begin{align*}
\mathbb{E}[\mathcal{E} | \mathcal{G}] & = \sum_{t = 1}^T{\mathbb{E}[\mathcal{E}_t | \mathcal{G}]}\\
& = \sum_{F \in \mathcal{G}}{\mathbb{E}[\mathcal{E} | F] \cdot \Pr[F | \mathcal{G}]}.
\end{align*}
We emphasize that $\mathbb{E}[\mathcal{E}_t | \mathcal{G}]$ and $\mathbb{E}[\mathcal{E} | \mathcal{G}]$ themselves are not necessarily to be relative entropies, although $\mathbb{E}[\mathcal{E}_t | F]$ and $\mathbb{E}[\mathcal{E}| F]$ are relative entropies for each single $F \in \mathcal{F}$.
	
Let $\mathcal{F}_0 \subseteq \mathcal{F}$ be the universe of all the ``legal'' configurations (i.e., the configurations with at most $R$ rounds in \textit{Setting~4}). We will show that $\mathbb{E}[\mathcal{E} | \mathcal{F}_0] = \Omega(D + R)$. We can then conclude that there exists one ``legal'' configuration $F \in \mathcal{F}_0$, such that $\ln{\frac{\Pr[b | F]}{\Pr'[b | F]}} = \Omega(D + R)$ for some recommendation sequence $b \in \mathcal{B}_1 \times \dotsb \times \mathcal{B}_T$.
	
Define $I_t^{(1.a)}$, $I_t^{(1.b)}$, $I_t^{(2.a)}$, $I_t^{(2.b)}$ to be the indicator random variables such that
\begin{itemize}
\item $I_t^{(1.a)} = 1$ iff the $t$th round is in \textit{Case~1.a};
\item $I_t^{(1.b)} = 1$ iff the $t$th round is in \textit{Case~1.b};
\item $I_t^{(2.a)} = 1$ iff the $t$th round is in \textit{Case~2.a};
\item $I_t^{(2.b)} = 1$ iff the $t$th round is in \textit{Case~2.b}.
\end{itemize}
Since the expected loss of $\mathcal{A}(C, \mathcal{V}\langle\mathcal{U}\rangle)$ is at least $\frac{1}{4}$ if the $t$th round is in \textit{Case~1.a} and \textit{Setting~1} is chosen (i.e., $\mathbb{E}[\mathcal{L}_t | I_t^{(1.a)} = 1 \text{ and } X_t = 1] \ge \frac{1}{4}$), we have
\begin{align}\label{equation:lower-bound-equation-1}
\mathbb{E}[\mathcal{L}_t] & \ge \frac{1}{4} \cdot \Pr[I_t^{(1.a)} = 1, X_t = 1]\notag\\
& = \frac{1}{4} \cdot \Pr[I_t^{(1.a)} = 1] \cdot \Pr[X_t = 1]\notag\\
& = \frac{1}{8}\left(1 - \frac{D + R}{2T}\right) \cdot \Pr[I_t^{(1.a)} = 1]\notag\\
& \ge \frac{1}{16} \cdot \mathbb{E}[I_t^{(1.a)}].
\end{align}
The second step is because $I_t^{(1.a)}$ and $X_t$ are independent (although $I_t^{(1.a)}$ may depend on $X_s$ for $s < t$). Take summation over $t$ on both sides of \eqref{equation:lower-bound-equation-1}, we get
\begin{equation}\label{equation:lower-bound-equation-2}
\mathbb{E}\left[\sum_{t = 1}^T{\mathcal{L}_t}\right] \ge \frac{1}{16} \cdot \mathbb{E}\left[\sum_{t = 1}^T{I_t^{(1.a)}}\right].
\end{equation}
Similarly, since the expected loss of $\mathcal{A}(C, \mathcal{V}\langle\mathcal{U}\rangle)$ is at least $\frac{1}{4}$ if the $t$th round is in \textit{Case~1.b} and \textit{Setting~2} is chosen (i.e., $\mathbb{E}[\mathcal{L}_t | I_t^{(1.b)} = 1 \text{ and } X_t = 2] \ge \frac{1}{4}$), we have
\begin{equation}\label{equation:lower-bound-equation-3}
\mathbb{E}\left[\sum_{t = 1}^T{\mathcal{L}_t}\right] \ge \frac{1}{16} \cdot \mathbb{E}\left[\sum_{t = 1}^T{I_t^{(1.b)}}\right].
\end{equation}
Combine \eqref{equation:lower-bound-equation-2} and \eqref{equation:lower-bound-equation-3}, we get
\begin{equation}\label{equation:lower-bound-equation-4}
\mathbb{E}\left[\sum_{t = 1}^T{\mathcal{L}_t}\right] \ge \frac{1}{32} \cdot \left(\mathbb{E}\left[\sum_{t = 1}^T{I_t^{(1.a)}}\right] + \mathbb{E}\left[\sum_{t = 1}^T{I_t^{(1.b)}}\right]\right).
\end{equation}
	
Recall that the expectation of $\mathcal{E}_t$ is at least $0.13$ if the $t$th round is in \textit{Case~2.a} and \textit{Setting~2} is chosen (i.e., $\mathbb{E}[\mathcal{E}_t | I_t^{(2.a)} = 1 \text{ and } X_t = 2] \ge 0.13$), and that expectation is at least $0$ in other cases (because the relative entropy is always non-negative), we have
\begin{align}\label{equation:lower-bound-equation-5}
\mathbb{E}[\mathcal{E}_t | \mathcal{F}_0] & \ge 0.13 \cdot \Pr[I_t^{(2.a)} = 1, X_t = 2 | \mathcal{F}_0]\notag\\
& \ge 0.13 \cdot \left(\Pr[I_t^{(2.a)} = 1, X_t = 2] - \frac{1}{T}\right)\notag\\
& = 0.13 \cdot \left(\Pr[I_t^{(2.a)} = 1] \cdot \Pr[X_t = 2] - \frac{1}{T}\right)\notag\\
& = 0.13 \cdot \left(\frac{1}{2}\left(1 - \frac{D + R}{2T}\right) \cdot \Pr[I_t^{(2.a)} = 1] - \frac{1}{T}\right)\notag\\
& \ge 0.13 \cdot \left(\frac{1}{4} \cdot \mathbb{E}[I_t^{(2.a)}] - \frac{1}{T}\right).
\end{align}
The second step is because
\begin{align*}
& \quad\ \Pr[I_t^{(2.a)} = 1, X_t = 2 | \mathcal{F}_0]\\
& = \frac{1}{\Pr[\mathcal{F}_0]}(\Pr[I_t^{(2.a)} = 1, X_t = 2] - \Pr[I_t^{(2.a)} = 1, X_t = 2 | \mathcal{F} \setminus \mathcal{F}_0] \cdot \Pr[\mathcal{F} \setminus \mathcal{F}_0])\\
& \ge \Pr[I_t^{(2.a)} = 1, X_t = 2] - \Pr[\mathcal{F} \setminus \mathcal{F}_0]\\
& \ge \Pr[I_t^{(2.a)} = 1, X_t = 2] - \frac{1}{T},
\end{align*}
and the third step is because $I_t^{(2.a)}$ and $X_t$ are independent (although $I_t^{(2.a)}$ may depend on $X_s$ for $s < t$). Take summation over $t$ on both sides of \eqref{equation:lower-bound-equation-5}, we get
\begin{align}\label{equation:lower-bound-equation-6}
\mathbb{E}[\mathcal{E} | \mathcal{F}_0] & = \mathbb{E}\left[\sum_{t = 1}^T{\mathcal{E}_t} \bigg| \mathcal{F}_0\right]\notag\\
& \ge 0.13 \cdot \left(\frac{1}{4} \cdot \mathbb{E}\left[\sum_{t = 1}^T{I_t^{(2.a)}}\right] - 1\right).
\end{align}
	
Similarly, by the fact that the expectation of $\mathcal{E}_t$ is at least $0.13$ if the $t$th round is in \textit{Case~2.b} and \textit{Setting~3} or \textit{Setting~4} is chosen (i.e., $\mathbb{E}[\mathcal{E}_t | I_t^{(2.b)} = 1 \text{ and } X_t = 0] \ge 0.13$), we have
\begin{align}\label{equation:lower-bound-equation-7}
\mathbb{E}[\mathcal{E}_t | \mathcal{F}_0] & \ge 0.13 \cdot \Pr[I_t^{(2.b)} = 1, X_t = 0 | \mathcal{F}_0]\notag\\
& \ge 0.13 \cdot \left(\Pr[I_t^{(2.b)} = 1] \cdot \Pr[X_t = 0] - \frac{1}{T}\right)\notag\\
& = 0.13 \cdot \left(\frac{D + R}{2T} \cdot \Pr[I_t^{(2.b)} = 1] - \frac{1}{T}\right)\notag\\
& = 0.13 \cdot \left(\frac{D + R}{2T} \cdot \mathbb{E}[I_t^{(2.b)}] - \frac{1}{T}\right).
\end{align}
Take summation over $t$ on both sides of \eqref{equation:lower-bound-equation-7}, we obtain
\begin{equation}\label{equation:lower-bound-equation-8}
\mathbb{E}[\mathcal{E} | \mathcal{F}_0] \ge 0.13 \cdot \left(\frac{D + R}{2T} \cdot \mathbb{E}\left[\sum_{t = 1}^T{I_t^{(2.b)}}\right] - 1\right).
\end{equation}
	
We are now ready to derive a lower bound for $\mathbb{E}[\mathcal{E} | \mathcal{F}_0]$:
\begin{itemize}
\item If $\mathbb{E}\left[\sum_{t = 1}^T{I_t^{(2.b)}}\right] \ge \frac{T}{2}$, by \eqref{equation:lower-bound-equation-8}, we have
\begin{align*}
\mathbb{E}[\mathcal{E} | \mathcal{F}_0] & \ge 0.13 \cdot \left(\frac{1}{4}(D + R) - 1\right)\\
& = \Omega(D + R).
\end{align*}
\item If $\mathbb{E}\left[\sum_{t = 1}^T{I_t^{(2.b)}}\right] < \frac{T}{2}$, then
\begin{equation*}
\mathbb{E}\left[\sum_{t = 1}^T{I_t^{(1.a)}}\right] + \mathbb{E}\left[\sum_{t = 1}^T{I_t^{(1.b)}}\right] + \mathbb{E}\left[\sum_{t = 1}^T{I_t^{(2.a)}}\right] > \frac{T}{2}.
\end{equation*}
By \eqref{equation:lower-bound-equation-4}, we have
\begin{align*}
\mathbb{E}\left[\sum_{t = 1}^T{I_t^{(2.a)}}\right] & > \frac{T}{2} - \left(\mathbb{E}\left[\sum_{t = 1}^T{I_t^{(1.a)}}\right] + \mathbb{E}\left[\sum_{t = 1}^T{I_t^{(1.b)}}\right]\right)\\
& \ge \frac{T}{2} - 32 \cdot \mathbb{E}\left[\sum_{t = 1}^T{\mathcal{L}_t}\right].
\end{align*}
Substitute this to \eqref{equation:lower-bound-equation-6}, we get
\begin{align*}
\mathbb{E}[\mathcal{E} | \mathcal{F}_0] & \ge 0.13 \cdot \left(\frac{1}{4} \cdot \left(\frac{T}{2} - 32 \cdot \mathbb{E}\left[\sum_{t = 1}^T{\mathcal{L}_t}\right]\right) - 1\right)\\
& \ge 0.13 \cdot \left(\frac{1}{4} \cdot \left(\frac{T}{2} - o(T)\right) - 1\right)\\
& = \Omega(T)\\
& = \Omega(D + R).
\end{align*}
The second step is because
\begin{align*}
\mathbb{E}\left[\sum_{t = 1}^T{\mathcal{L}_t}\right] & = \mathbb{E}\left[\sum_{t = 1}^T{\mathcal{L}_t} \bigg| \mathcal{F}_0\right] \cdot \Pr[\mathcal{F}_0] + \mathbb{E}\left[\sum_{t = 1}^T{\mathcal{L}_t} \bigg| \mathcal{F} \setminus \mathcal{F}_0\right] \cdot \Pr[\mathcal{F} \setminus \mathcal{F}_0]\\
& \le \mathbb{E}\left[\sum_{t = 1}^T{\mathcal{L}_t} \bigg| \mathcal{F}_0\right] + T \cdot \frac{1}{T}\\
& \le \max_{F \in \mathcal{F}_0}{\mathbb{E}\left[\sum_{t = 1}^T{\mathcal{L}_t} \bigg| F\right]} + 1\\
& \le L + 1\\
& = o(T).
\end{align*}
\end{itemize}
	
We have shown that Theorem~\ref{theorem:lower-bound-privacy-general-setting} holds if $6\ln{T} \le D + R \le T$ and $P = 1$. Now we prove Theorem~\ref{theorem:lower-bound-privacy-general-setting} for the cases where $D + R > T$ and/or $P > 1$.
	
\paragraph{When $D + R > T$ (still assuming $P = 1$).} When $D + R > T$, we can choose another pair of $D'$ and $R'$, such that $D' \le R$, $R' \le R$ and $D' + R' = T$. If we use $D'$ and $R'$ (instead of $D$ and $R$) to do the above construction, we can show that the $\mathbb{E}[\mathcal{E} | \mathcal{F}_0] = \Omega(D' + R')$. Recall that $D \in \{0, \dotsc, T\}$ and $R \in \{0, \dotsc, T\}$, thus $D + R \le 2T$ and $\Omega(D' + R') = \Omega(D + R)$. This finishes the proof.
	
\paragraph{When $P > 1$.} In the above proof for the case of $P = 1$, we construct two voting patterns $\mathcal{V}\langle\mathcal{U}\rangle$ and $\mathcal{V}\langle\mathcal{U}'\rangle$, each of which contains a single voter $U$ and $U'$, respectively. We can extend this proof to the case of $P > 1$ by considering a voting pattern containing $P$ voters who vote exactly the same as $U$, and another voting pattern containing $P$ voters who vote exactly the same as $U'$. Such a pair of voting patterns are $2P$-step adjacent to each other, and we will get an additional factor of $\Theta(\frac{1}{P})$ in our lower bound for the $\epsilon$.
\end{proof}

\subsection{Algorithm}\label{section:p-REC}
We propose the following p-REC algorithm for the general setting where $D + R \ge 0$. The p-REC algorithm is a generalized version of the p-REC$_\text{sim}$ algorithm, and it shares a similar structure as that of p-REC$_\text{sim}$, except that the procedure \texttt{UpdateWeight()} is replaced by \texttt{UpdateCreditAndWeight()}. In fact, we can get back the p-REC$_\text{sim}$ algorithm by setting $D = R = 0$ in the p-REC algorithm.
\begin{algorithm}[h]
	\SetKwInOut{Input}{Input}
	\SetKwInOut{Output}{Output}
	\SetKwInOut{Initialization}{Initialization}
	\SetKwData{Weight}{weight}
	\SetKwData{DCredit}{credit$^{(D)}$}
	\SetKwData{RCredit}{credit$^{(R)}$}
	\SetKwData{Obj}{obj}
	\SetKwData{Feedback}{feedback}
	\SetKwFunction{Main}{Main}
	\SetKwFunction{RecommendByWeight}{RecommendByWeight}
	\SetKwFunction{UpdateCreditAndWeight}{UpdateCreditAndWeight}
	
	\Input{A client $C$, a voting pattern $\mathcal{V}\langle\mathcal{U}\rangle$}
	\Output{Recommend an object from $\mathcal{B}_t$ to client $C$ in each round $t$}
	\Initialization{$\gamma \gets \frac{m}{(3T / (R + 1)) - 1}$; $\lambda \gets 2m\ln\frac{T}{R + 1}$; $\rho \gets \frac{1}{2m}$; for each $i \in \mathcal{U}$: $\DCredit[i] \gets 2D$, $\RCredit[i] \gets 2R + 1$, $\Weight[i] \gets 1$}
	
	\BlankLine
	
	\SetKwBlock{MainProcedure}{Procedure \Main{}}{end}
	\MainProcedure{
		\ForEach{round $t = 1, \dotsc, T$}{
			$\Obj \gets$ \RecommendByWeight{$\Weight[\;]$}\;
			Recommend object $\Obj$ to the client $C$\;
			$\Feedback \gets$ the client $C$'s feedback on object $\Obj$\;
			\UpdateCreditAndWeight{$\DCredit[\;]$, $\RCredit[\;]$, $\Weight[\;]$, $\Obj$, $\Feedback$}\;
		}
	}
	
	\SetKwBlock{UpdateCreditAndWeightProcedure}{Procedure \UpdateCreditAndWeight{$\DCredit[\;]$, $\RCredit[\;]$, $\Weight[\;]$, $\Obj$, $\Feedback$}}{end}
	\UpdateCreditAndWeightProcedure{
		\eIf{$\Feedback =$ ``dislike''}{
			$\RCredit[i] \gets \RCredit[i] - 1$ for every voter $i$ who votes on $\Obj$\;
		}
		{
			$\DCredit[i] \gets \DCredit[i] - 1$ for every voter $i$ who does not vote on $\Obj$\;
		}
		\ForEach{voter $i \in \mathcal{U}$}{
			\eIf{$\RCredit[i] > 0$ and $\DCredit[i] + \RCredit[i] > 0$}{
				$\Weight[i] \gets 1$\;
			}
			{
				$\Weight[i] \gets 0$\;
			}
		}
	}
	
	\caption{Privacy-preserving RECommendation (p-REC) algorithm.}\label{algorithm:p-REC}
\end{algorithm}

In the beginning of the p-REC algorithm, each voter $i \in \mathcal{U}$ is initialized with two \emph{credit} values $\mathsf{credit}^{(D)}[i] = 2D$ (which we call \emph{$D$-credit}) and $\mathsf{credit}^{(R)}[i] = 2R + 1$ (which we call \emph{$R$-credit}). In each round $t$, the algorithm recommends objects in the same way as the p-REC$_\text{sim}$ algorithm by invoking the \texttt{RecommendByWeight()} procedure. After it receives the client's feedback, the algorithm updates each voter's credit and then calculate his/her weight by invoking the \texttt{UpdateCreditAndWeight()} procedure.

To see the intuition behind the p-REC algorithm, let us analyze why the p-REC$_\text{sim}$ algorithm fails in the general setting where $D + R \ge 0$. If we run p-REC$_\text{sim}$ in the general setting, we may end up with a situation where all the client's peers are kicked out from the system. A client's peer can be (wrongly) ``kicked out'' in two scenarios:
\begin{itemize}
\item when the client likes more than one objects in some round, the peer votes on one such object, but another such object is recommended;
\item when the peer votes on an object that the client dislikes, and that object is recommended to the client.
\end{itemize}
However, since these two scenarios can happen for at most $D + R$ times, a natural idea is to give a voter $D + R$ more ``chances'' before we kick out him/her. Motivated by this, we could initialize each voter $i$ with $D + R + 1$ credit, and deduct $i$'s credit by $1$ when $i$ is caught to vote on an object the client dislikes, or when the client likes the recommended object but $i$ does not vote on it. We kick out a voter only when he/she has no credit.

For some technical reasons that will be clear later, the p-REC algorithm needs to introduce two types of credit ($D$-credit and $R$-credit), and deducts different types of credit in different situations. It also initializes each voter with $2D$ (instead of $D$) $D$-credit and $2R + 1$ (instead of $R + 1$) $R$-credit.

\paragraph{Loss analysis of the p-REC algorithm.} In the remaining parts of this section, we prove Theorem~\ref{theorem:p-REC-loss-and-privacy}. We first analyze the loss of the p-REC algorithm, we have:
\begin{theorem}
The p-REC algorithm guarantees $O((R + 1)\log{\frac{n}{P}})$ worst-case expected loss, even for an adaptive adversary.
\end{theorem}
\begin{proof}
The proof of this theorem is very similar to the proof of Theorem~\ref{theorem:upper-bound-loss-simple-setting}. First consider all the rounds with $Z_t = 0$. Let $W_t$ be the number of surviving voters in the $t$th round and $C_t^{(R)} = \sum_{i \in \mathcal{U}}{\mathsf{weight}[i] \cdot \mathsf{credit}^{(R)}[i]}$ be the total $R$-credit of all these $W_t$ surviving voters in the $t$th round. If the client incurs a loss in the $t$th round, at least $\rho$ fraction of surviving voters will get their $R$-credit deducted by $1$. That is,
\begin{equation}\label{equation:upper-bound-loss-general-setting}
C_{t + 1}^{(R)} \le C_t^{(R)} - \rho W_t.
\end{equation}
Since each voter has at most $2R + 1$ $R$-credit, we have $C_t^{(R)} \le (2R + 1)W_t$. Substitute this back to \eqref{equation:upper-bound-loss-general-setting}, we get
\begin{equation*}
C_{t + 1}^{(R)} \le C_t^{(R)} \cdot \left(1 - \frac{\rho}{2R + 1}\right).
\end{equation*}
In other words, $C_t^{(R)}$ decreases by at least $\frac{\rho}{2R + 1}$ fraction whenever the client incurs a loss. Notice that the total $R$-credit of all the surviving voters are initially at most $(2R + 1)n$ and always at least $(R + 1)P$ (the client's $P$ peers always have at least $R + 1$ $R$-credit), it then follows that the number of loss rounds with $Z_t = 0$ can be at most $\frac{2R + 1}{\rho}\ln{\frac{(2R + 1)n}{(R + 1)P}} = O((R + 1)\log{\frac{n}{P}})$.

For those rounds with $Z_t = 1$, they cause at most $\gamma T = O(R)$ additional expected loss.
\end{proof}
Here we see the reason of introducing two types of credit: if we only had one type of credit, the upper bound of loss would be $O((D + R + 1)\log{\frac{n}{P}})$, which is not only linear to $R$, but also to $D$ --- this is logically wrong, because if a client likes more objects, the loss should only be smaller.

\paragraph{Privacy analysis of the p-REC algorithm.} Next we analyze the privacy of the p-REC algorithm, we have:
\begin{theorem}\label{theorem:upper-bound-privacy-general-setting}
The p-REC algorithm preserves $O(\frac{1}{P}(D + R + 1)\log{\frac{T}{R + 1}})$-differential privacy, even for an adaptive adversary.
\end{theorem}
The proof of Theorem~\ref{theorem:upper-bound-privacy-general-setting} follows a very similar structure to that of Theorem~\ref{theorem:upper-bound-privacy-simple-setting}. Again, we consider two executions p-REC$(C, \mathcal{V}\langle\mathcal{U}\rangle)$ and p-REC$(C, \mathcal{V}\langle\mathcal{U}'\rangle)$ for a fixed client $C$ and a fixed pair of adjacent voting patterns $\mathcal{V}\langle\mathcal{U}\rangle, \mathcal{V}\langle\mathcal{U}'\rangle$ ($\mathcal{U}$ contains one more voter than $\mathcal{U}'$). Our goal is to show that $|\mathcal{E}(b)| = O(\frac{1}{P}(D + R + 1)\log{\frac{T}{R + 1}})$ for any fixed recommendation sequence $b = (b_1, \dotsc, b_T) \in \mathcal{B}_1 \times \dotsb \times \mathcal{B}_T$.

We define the number of surviving voters $W_t(b)$ in the same way as before. In addition, we define $C_t(b) = \sum_{i \in \mathcal{U}}{\mathsf{weight}[i] \cdot (\mathsf{credit}^{(D)}[i] + \mathsf{credit}^{(R)}[i])}$ to be the total credit ($D$-credit plus $R$-credit) of all the surviving voters at the beginning of round $t$ in the execution $\text{p-REC}(C, \mathcal{V}\langle\mathcal{U}\rangle)$, conditioned on that the recommendations in the first $t - 1$ rounds are $b_1, \dotsc, b_{t - 1}$. We have:
\begin{lemma}\label{lemma:single-round-privacy-leakage-general-setting}
For any round $t$, $|\mathcal{E}_t(b)| \le 3(2D + 2R + 1) \cdot \frac{\lambda}{C_t(b)}$.
\end{lemma}
\begin{proof}
We can show $|\mathcal{E}_t(b)| \le 3\lambda \cdot \frac{1}{W_t(b)}$ in exactly the same way as the proof of Lemma~\ref{lemma:single-round-privacy-leakage-simple-setting}. Recall that each voter has at most $2D + 2R + 1$ total credit, we therefore have $C_t(b) \le (2D + 2R + 1)W_t(b)$. Lemma~\ref{lemma:single-round-privacy-leakage-general-setting} is then proved by substituting this back to $|\mathcal{E}_t(b)| \le 3\lambda \cdot \frac{1}{W_t(b)}$.
\end{proof}
\begin{lemma}\label{lemma:round-with-non-zero-privacy-leakage-general-setting}
For any round $t$, if $|\mathcal{E}_t(b)| \neq 0$, then
\begin{equation*}
C_{t + 1}(b) \le C_t(b) \cdot \left(1 - \frac{1}{3m(2D + 2R + 1)}\right).
\end{equation*}
\end{lemma}
\begin{proof}
By following the same arguments in the proof of Lemma~\ref{lemma:round-with-non-zero-privacy-leakage-simple-setting}, we can show that at least $\frac{1}{3m}$ fraction of surviving voters get their $D$-credit or $R$-credit deducted by $1$ after the $t$th round. That is,
\begin{equation}\label{equation:upper-bound-privacy-general-setting}
C_{t + 1}(b) \le C_t(b) - \frac{1}{3m}W_t(b).
\end{equation}
Since $C_t(b) \le (2D + 2R + 1)W_t(b)$, substituting this back to \eqref{equation:upper-bound-privacy-general-setting}, we have
\begin{equation*}
C_{t + 1}(b) \le C_t(b) \cdot \left(1 - \frac{1}{3m(2D + 2R + 1)}\right),
\end{equation*}
as desired.
\end{proof}
Now we can prove Theorem~\ref{theorem:upper-bound-privacy-general-setting}:
\begin{proof}[Proof of Theorem~\ref{theorem:upper-bound-privacy-general-setting}]
Let $t_1, \dotsc, t_K$ be the rounds in which $|\mathcal{E}_{t_i}(b)| \neq 0$.
By replacing Lemma~\ref{lemma:single-round-privacy-leakage-simple-setting} and Lemma~\ref{lemma:round-with-non-zero-privacy-leakage-simple-setting} with Lemma~\ref{lemma:single-round-privacy-leakage-general-setting} and Lemma~\ref{lemma:round-with-non-zero-privacy-leakage-general-setting}, respectively, and then following the same arguments in the proof of Theorem~\ref{theorem:upper-bound-privacy-simple-setting}, we can show that
\begin{equation}\label{equation:upper-bound-privacy-general-setting-1}
|\mathcal{E}(b)| \le 3(2D + 2R + 1)\lambda\sum_{i = 1}^{K}{\frac{1}{C_{t_i}(b)}}
\end{equation}
and
\begin{equation}\label{equation:upper-bound-privacy-general-setting-2}
\sum_{i = 1}^{K}{\frac{1}{C_{t_i}(b)}} \le 3m(2D + 2R + 1) \cdot \frac{1}{C_{t_K}(b)}.
\end{equation}
Since $C_{t_K}(b) \ge P(D + R + 1)$, substitute this to \eqref{equation:upper-bound-privacy-general-setting-2}, we have
\begin{equation}\label{equation:upper-bound-privacy-general-setting-3}
\sum_{i = 1}^{K}{\frac{1}{C_{t_i}(b)}} \le 3m(2D + 2R + 1) \cdot \frac{1}{P(D + R + 1)} \le 6m \cdot \frac{1}{P}.
\end{equation}
We get the desired result by substitute \eqref{equation:upper-bound-privacy-general-setting-3} back to \eqref{equation:upper-bound-privacy-general-setting-1}.
\end{proof}
Here the reason of giving each voter $2D$ initial $D$-credit and $2R + 1$ initial $R$-credit becomes clear: if we were to give each voter $D$ initial $D$-credit and $R + 1$ initial $R$-credit, we could only get $C_{t_K}(b) \ge P$ instead of $C_{t_K}(b) \ge P(D + R + 1)$, which would introduce an extra $\Theta(D + R)$ factor in the result.


\end{document}